\documentclass[lettersize,journal,twoside]{IEEEtran}
\usepackage{cite}\usepackage{amsmath,amsthm,amssymb,amsfonts}
\usepackage{graphicx}
\usepackage{algorithm,algorithmic}
\usepackage{textcomp}
\usepackage{mathrsfs}
\usepackage{tcolorbox}
\usepackage{url}
\usepackage{subcaption}
\usepackage{multirow}
\usepackage{float}
\floatname{algorithm}{}       

\newtheorem{theorem}{Theorem}
\newtheorem{corollary}{Corollary}
\newtheorem{lemma}{Lemma}

\newtheorem{definition}{Definition}
\newtheorem{assumption}{Assumption}
\newtheorem{remark}{Remark}

\newenvironment{optimization}[1]{
    \begin{tcolorbox}[colback=white,colframe=black,arc=0mm,title={ #1},left=0pt,right=0pt,top=0pt,bottom=0pt]
    }{
    \end{tcolorbox}
}

\usepackage{hyperref}
\usepackage{fancyhdr}
\fancypagestyle{arxiv}{%
  \fancyhf{} % 清空默认
  \fancyhead[C]{%
    This is the author’s version of the article accepted for publication in \textit{IEEE Transactions on Power Systems}%
  }
  \fancyfoot[C]{%
    Final version available at \href{https://ieeexplore.ieee.org/abstract/document/11169504}{doi:10.1109/TPWRS.2025.3611641}%
  }
}

\begin{document}

\title{Frequency-Varying Optimization: A Control Framework for New Dynamic Frequency Response Services}

\author{Yiqiao Xu, \IEEEmembership{Member, IEEE}, Quan Wan, \IEEEmembership{Graduate Student Member, IEEE}, and Alessandra Parisio, \IEEEmembership{Senior Member, IEEE}
\thanks{This work was partially supported by Supergen Energy Networks Impact Hub (EP/Y016114/1) and Grid Scale Thermal and Thermo-Chemical Electricity Storage (EP/W027860/1). Yiqiao Xu, Quan Wan, and Alessandra Parisio are with the Department of Electrical and Electronic Engineering, The University of Manchester, Manchester, M13 9PL U.K. (e-mails: yiqiao.xu@manchester.ac.uk, alessandra.parisio@manchester.ac.uk).}
\thanks{\noindent
\copyright~2025 IEEE. Personal use of this material is permitted. 
Permission from IEEE must be obtained for all other uses, in any current or future media, 
including reprinting/republishing this material for advertising or promotional purposes, 
creating new collective works, for resale or redistribution to servers or lists, 
or reuse of any copyrighted component of this work in other works.  
The final version of record is available at:  
\href{https://ieeexplore.ieee.org/abstract/document/11169504}{doi:10.1109/TPWRS.2025.3611641}.
}
}

\markboth{Journal of \LaTeX\ Class Files,~Vol.~14, No.~8, August~2021}%
{Xu \MakeLowercase{\textit{et al.}}: Frequency-Varying Optimization: A Control Framework for New Dynamic Frequency Response Services}

\IEEEpubid{0000--0000/00\$00.00~\copyright~2021 IEEE}

\maketitle

\thispagestyle{arxiv}   % 首页页脚加声明
\pagestyle{arxiv}       % 后续页也加声明
\begin{abstract}
     To address the variability of renewable generation, initiatives have been launched globally to provide faster and more effective frequency responses. In the UK, the National Energy System Operator (NESO) has introduced a suite of three new dynamic services, where aggregation of assets is expected to play a key role. For an Aggregated Response Unit (ARU), the required level of frequency response varies with grid frequency, resulting in a frequency-varying equality constraint that assets should meet collectively. We show that the optimal coordination of an ARU constitutes a Frequency-Varying Optimization (FVO) problem, in which the optimal trajectory for each asset evolves dynamically. To facilitate online optimization, we reformulate the FVO problem into Tracking of the Optimal Trajectory (TOT) problems, with algorithms proposed for two scenarios: one where the asset dynamics are negligible, and another where they must be accounted for. Under reasonable conditions, the ARU converges to the optimal trajectory within a fixed time, and within the maximum delivery time requested by NESO. The proposed framework can be readily distributed to coordinate a large number of assets. Numerical results verify the effectiveness and scalability of the proposed control framework.
\end{abstract}

\begin{IEEEkeywords}
Frequency-varying optimization, fast frequency response, optimization algorithm.
\end{IEEEkeywords}

\section{Introduction}
\IEEEPARstart{T}{he} increasing penetration of renewable generation has led to reduced system inertia, posing significant challenges to frequency stability \cite{9796617}. Given the limited interconnections with continental grids, many island countries have established Fast Frequency Response (FFR) services to mitigate the impact of low inertia \cite{8864014}. In recent years, the value of energy storage systems in delivering FFR services has been well demonstrated \cite{pusceddu2021synergies, li2022variable}, particularly for Battery Energy Storage Systems (BESSs), which stand out for their fast response capabilities and operational flexibility \cite{akram2020review, wu2024dynamic, cao2024battery}, with applications spanning the generation, grid, and user sides.

As part of the UK's FFR services, firm frequency response, encompassing both static and dynamic types, is procured to manage frequency deviations in real time \cite{NGESO_FFR}. In response to evolving system needs, the National Energy System Operator (NESO) has phased out Dynamic Firm Frequency Response (DFFR) in favor of a suite of three new dynamic services---Dynamic Regulation (DR), Dynamic Moderation (DM), and Dynamic Containment (DC) \cite{future}. Designed to maintain grid frequency within 50$\pm$0.5 Hz, these services differ in their response envelopes and require participating units to respond rapidly and proportionally to frequency events, as illustrated in Fig. 1. DR is a pre-fault service for correcting continuous, small frequency deviations, with up to 10 s allowed for full delivery. DM is also a pre-fault service, primarily active during periods of high volatility, but must be fully delivered within 1 s. In contrast, DC is a post-fault service designed to arrest large frequency deviations, such as those caused by system faults, and is likewise required to respond within 1 s.

\begin{figure}[t]	\centerline{\includegraphics[width=0.48\textwidth]{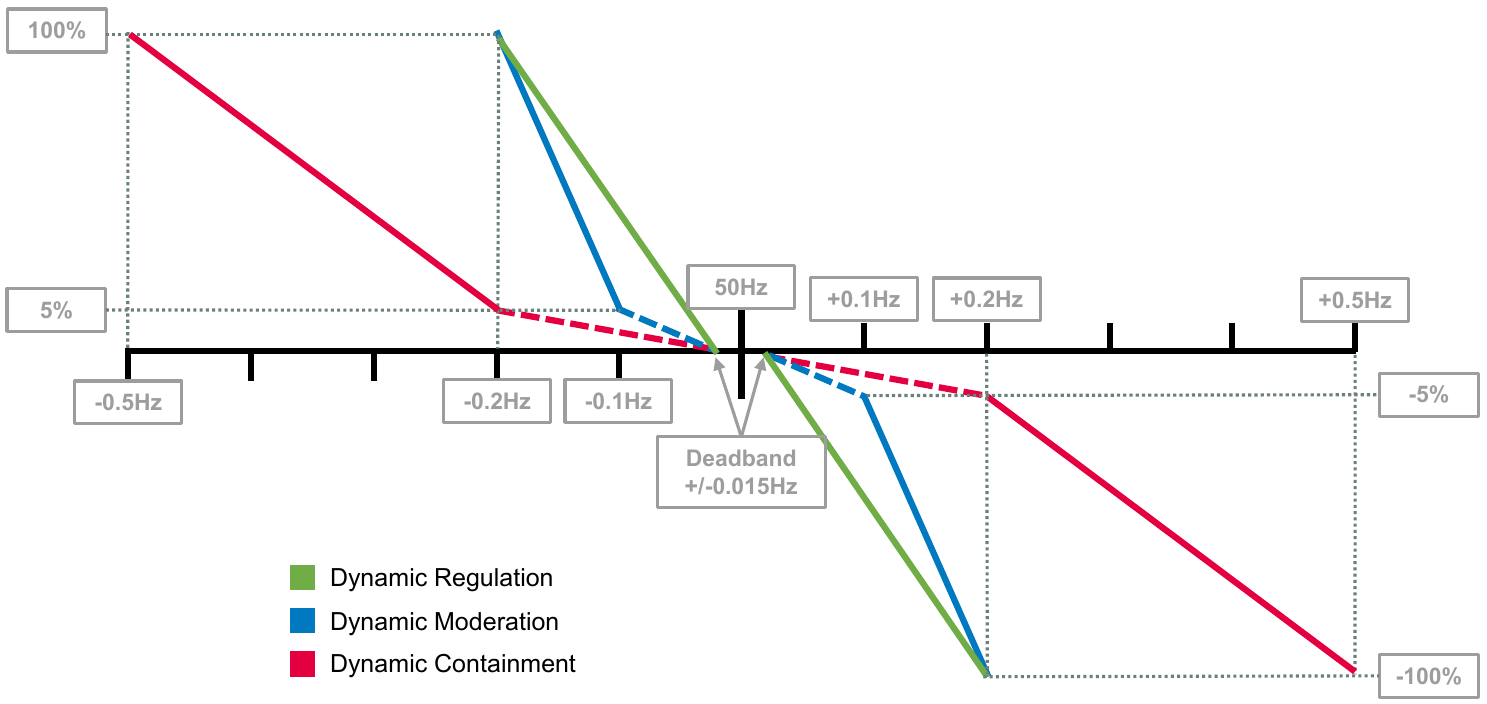}}
	\caption{Delivery requirement curves for DC, DM, and DR \cite{NGESO_guide}.}
	\label{DR section droop}
\end{figure}

\IEEEpubidadjcol

Assets at both the transmission and distribution levels are encouraged to participate. If eligible and located within the same control area, assets (generation, demand, or energy-limited) can register as an Aggregated Response Unit (ARU), an entity that is expected to become increasingly prevalent \cite{NGESO_guide}. However, a minimum response capacity of 1 MW is required per response unit, with maximum capacities of 100 MW for DC and 50 MW for DR and DM. To date, utility-scale storage projects have played a dominant role in these services, as evidenced by recent auction results \cite{NGESO_Auction}. These projects typically consist of multiple BESSs at one or more sites, and given the necessity of responding as contracted, synchronizing responses, and optimizing essential criteria, their aggregation should be conducted optimally. 

Accordingly, this paper investigates the optimal coordination of an ARU, where the responses of BESSs must collectively meet the delivery requirement curve, comply with state/input constraints, and minimize the cost of deviating from their operational baselines submitted to NESO. Three major challenges are identified for the new dynamic services:
\begin{itemize}
\item The delivery requirement curve yields a frequency-varying equality constraint that all assets need to meet, causing the optimal trajectory for each asset to evolve dynamically;

\item The services need to be delivered within the maximum delivery time (10/1/1 s for DR/DM/DC), which imposes an upper bound on the time available for coordination to be optimized online.

\item The control framework should be computationally efficient and scalable to coordinate a large number of assets in real time.
\end{itemize}

Among existing optimization-based control frameworks, online convex optimization offers explicit convergence-rate guarantees and is relatively easy to implement for frequency regulation \cite{8999747,10148805}. To minimize a smooth convex function, the convergence rate is upper bounded by $\mathcal{O}(1/\sqrt{T})$ \cite{hazan2016introduction}, which is insufficient for such nonstationary applications. Their continuous-time variants have also been employed for optimal frequency control \cite{7163587,wang2024distributed}; still, enhancements are necessary to achieve convergence within the maximum delivery time requested by NESO. In \cite{9107487}, a model predictive control-based supervisory layer is proposed to optimize converter setpoints, which are then tracked by a Proportional-Integral (PI) controller. In \cite{10820007}, a neural PI controller is developed to solve a steady-state economic dispatch problem and restore frequency, which, in case of need, may provide tracking guarantees through neural network design. In \cite{9744119, li2022coordinated, 9779512}, (deep) reinforcement learning algorithms are proposed for frequency regulation services. Similar to \cite{10820007}, these algorithms feature offline training and online decision-making but need to be retrained whenever the environment changes. In \cite{9399101}, a real-world demonstration of distributed control for heterogeneous assets is presented, solving a coordination problem at each regulation instant to optimally allocate the RegD signal in Automatic Generation Control (AGC). In \cite{8949561}, a hierarchical control framework comprising a central supervisor, regional aggregators, and local controllers is designed for FFR. To provide FFR within 500 ms, an online control framework is developed in \cite{9234642}, achieving fast and near-optimal coordination of thousands of assets across the distribution feeder.

\begin{figure}[t]
    \centering
    \begin{subfigure}{0.48\textwidth}
        \centering
        \includegraphics[width=\textwidth]{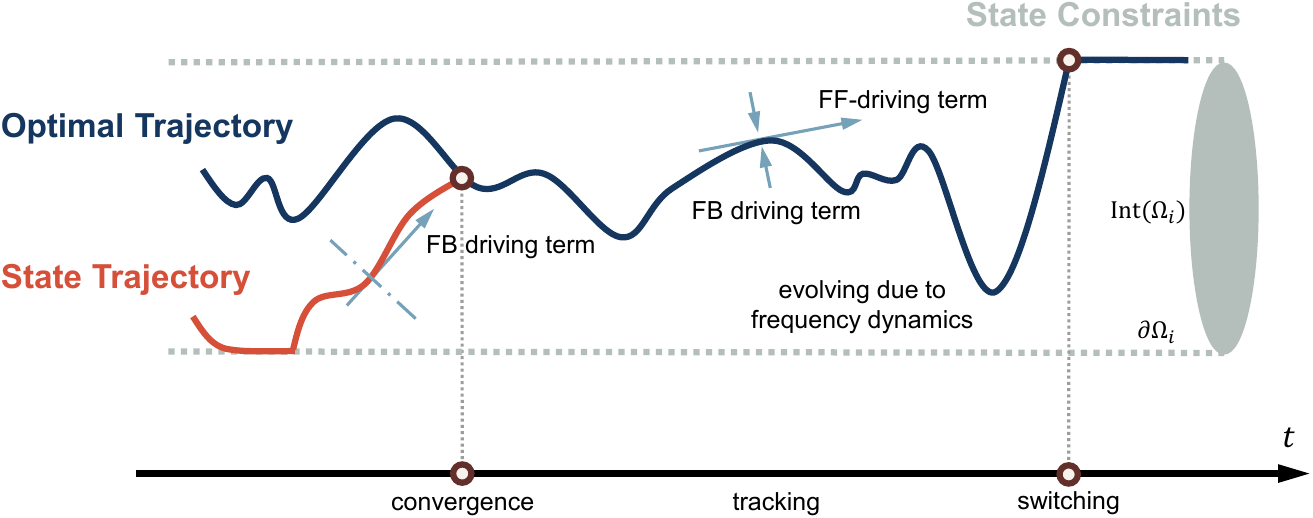}
    \end{subfigure}
    \begin{subfigure}{0.48\textwidth}
        \centering
        \includegraphics[width=\textwidth]{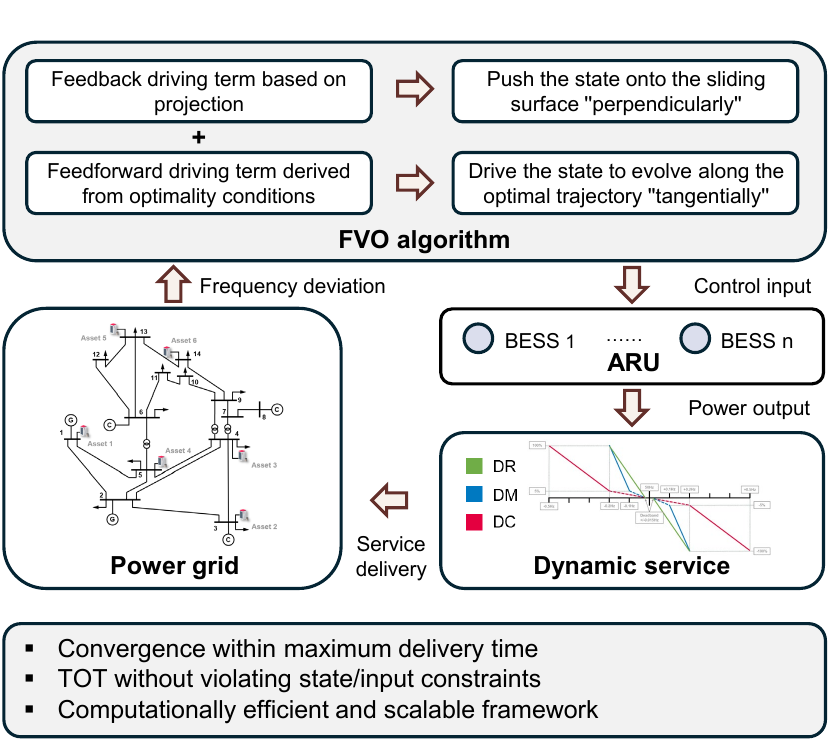}
    \end{subfigure}
    \caption{Overview of the proposed control framework.}
\end{figure}

Ensuring that the deployed resource consistently tracks the required quantity remains a challenge in \cite{7163587,wang2024distributed,8999747,10148805,8949561,9234642,9399101}, as these control frameworks are purely feedback-driven and lack predictive components. In terms of tracking, prediction-correction algorithms \cite{7902101,7993088,8062794} have been developed for Time-Varying Optimization (TVO) problems, where the cost and/or constraints are parameterized by time. Here, the goal is partially aligned with output regulation \cite{isidori1990output}, which concerns tracking a trajectory in the presence of disturbances generated by an exosystem, often through feedforward control \cite{ding2022distributed}. This alignment motivates us to introduce Frequency-Varying Optimization (FVO) as a specific class of TVO. We design feedforward driving terms to track the optimal trajectory with vanishing tracking errors. The variations in the trajectory are intrinsically tied to grid frequency dynamics and therefore interpreted from an exosystem perspective. The main contributions of this paper are as follows:
\begin{itemize}
    \item This paper bridges frequency response services with FVO. We show that optimal coordination of an ARU constitutes an FVO problem, where the equality constraint varies with grid frequency and the optimal trajectory evolves dynamically. While this paper targets the UK's new dynamic services, the methodology is directly applicable to international FFR services and supports both transmission- and distribution-level ARUs.
    
    \item For implementation purposes, the FVO problem is reformulated into two Tracking of the Optimal Trajectory (TOT) problems, without or with consideration of asset dynamics. Compared to Problem TOT-1, Problem TOT-2 addresses a more practical but challenging scenario where the asset dynamics are non-negligible.
    
    \item We propose a control framework termed FVO, featuring projected algorithms that incorporate fixed-time control \cite{8322314} and feedback/feedforward driving terms to find and track the optimal trajectory, as illustrated in Fig. 2. These algorithms solve the TOT problems online without violating state/input constraints, guaranteeing convergence within the maximum delivery time through appropriate selection of control parameters.

    \item The effectiveness and scalability of the proposed control framework are verified on the IEEE 14-bus and 39-bus systems in MATLAB/Simulink. Results show that the required quantity is delivered within the maximum delivery time requested by NESO, and that the overall response approaches the ideal outcomes achievable by DR/DM/DC. Furthermore, the control framework is computationally efficient for coordinating a large number of assets and adaptable to distributed implementation.
\end{itemize}

The remainder of this paper is organized as follows: 
Section \uppercase\expandafter{\romannumeral2} presents some preliminaries. Section \uppercase\expandafter{\romannumeral3} and \uppercase\expandafter{\romannumeral4} present our main results including the problems and algorithms. Section \uppercase\expandafter{\romannumeral5} demonstrates the effectiveness and scalability of the proposed framework through case studies. Section \uppercase\expandafter{\romannumeral6} concludes this paper. Finally, Section \uppercase\expandafter{\romannumeral7} provides the Appendix.

\section{Preliminaries}
\subsection{Notation, Terminology, and Basics}
The sets of real numbers and integers are denoted by $\mathbb{R}$ and $\mathbb{Z}$, respectively. We use $\mathbb{R}_{\geq0}$ and $\mathbb{R}_{>0}$ to denote the sets of non-negative and positive real numbers, and similarly for $\mathbb{Z}_{\geq0}$ and $\mathbb{Z}_{>0}$. A vector $x(t)$ is considered as the value of a variable at time $t$, unless otherwise specified. We denote by $x_i(t)$ the $i$-th element of the vector $x(t)$. 

\begin{definition}
    Suppose $x(t),y(t)\in\mathbb R^n$ are real-valued vectors where $n\in\mathbb Z_{>0}$. We write $x(t)\to^T y(t)$ if there exists a time $T\geq 0$ such that $\lim_{t\to T} x(t) = y(t)$ and $x(t) = y(t)$ $\forall t\geq T$.
\end{definition}

\begin{definition}
    The signum function is defined by $\text{sign}(x)=1$ if $x>0$, $\text{sign}(x)=0$ if $x=0$, and $\text{sign}(x)=-1$ otherwise. The sigmoid-like function is defined by $\text{sig}(x)=1$ if $x=0$ and $\text{sig}(x)=0$ otherwise.
\end{definition}

\begin{definition}
     Suppose $x\in\mathbb R^n$ and $n\in\mathbb Z_{>0}$. The projection of $x\in\mathbb R^n$ onto a closed convex set $\Omega\subseteq\mathbb{R}^n$ is defined as $\mathcal P_{\Omega}(x) = \arg\min_{y\in\Omega}\Vert x-y\Vert_2$. The projection is unique since $\Omega$ is closed and convex. 
\end{definition}

\begin{lemma}
Consider the indicator function
\[
\mathbb I(x)=
\begin{cases}
0,& x\in\Omega,\\
+\infty,& \text{otherwise},
\end{cases}
\]
and let $g(\epsilon,x):\mathbb R_{>0}\times\mathbb R\to\mathbb R$ be a function that is twice continuously differentiable and convex with respect to $x$, and satisfies $\lim_{\epsilon\to+\infty}g(\epsilon,x)=\mathbb I(x)$ point-wise. Then it holds that $\lim_{\epsilon\to+\infty}g_{xx}(\epsilon,x)=0$ if $x\in\mathrm{Int}(\Omega)$ and $\lim_{\epsilon\to+\infty}g_{xx}(\epsilon,x)=+\infty$ if $x\in\partial\Omega$, where $g_{xx}(\epsilon,x)$ denotes the second derivative of $g(\epsilon,x)$ with respect to $x$.
\end{lemma}
\begin{proof}
The proof is provided in Appendix A.
\end{proof}

\begin{lemma}\cite{boyd2004convex}
    Consider a closed convex set $\Omega\subset\mathbb R^n$. For all $x\in\mathbb R^n$ and $y\in\Omega$, the following inequality holds for the projection operation:
    \begin{align}
    \langle x-\mathcal P_{\Omega}(x),y-\mathcal P_{\Omega}(x)\rangle \leq 0.\notag
    \end{align}
\end{lemma}

\subsection{Grid Frequency Dynamics}
This section describes the grid frequency dynamics, which, in the context of this work, are interpreted as those of an exosystem. Consider a power system with an ARU composed of $n$ assets deployed at one or more points of common coupling. The frequency dynamics are described by the following set of equations \cite{kundur,6702462}:
\begin{align}
    \dot\delta_k(t) &= \Delta\omega_k(t) = \omega_k(t)-\omega^\star,\label{delta}\\
    \dot\omega_k(t) &= -\frac{\omega^\star}{2H_k}(D_k\Delta\omega_k(t) + \hat P_k(t) + \sum_{l\in\mathcal B_k}P_{kl}(t)),\label{omega}\\
    \hat P_k(t) &= P_{k,L}(t) - P_{k,G}(t) - \sum_{i\in\mathcal A_k}P_i(t),\label{net}\\
    P_{kl}(t) &= V_k(t)V_l(t)\vert Y_{kl}\vert\cos(\delta_k(t)-\delta_l(t)-\phi_{kl}),\label{flow}
\end{align}
where the subscripts $kl$ and $l$ refer to the $k$-th and $l$-th buses, respectively, and $i$ refers to the $i$-th asset.
Specifically, $\delta_k(t)$ is the phase angle at time $t$, $\Delta\omega_k(t)$ the frequency deviation, $\omega_k(t)$ the angular frequency, $\omega^\star$ the nominal frequency, and $V_k(t)$ the voltage magnitude. $H_k$ and $D_k$ represent the inertia and damping, respectively. $P_{k,G}(t)$ denotes the power generation, $P_{k,L}(t)$ the load demand, and $P_{i}(t)$ the power output of the $i$-th asset. $\mathcal A_k$ denotes the set of assets deployed at the $k$-th bus, and $\mathcal B_k$ denotes the set of buses connected to the $k$-th bus. The power flow from the $k$-th bus to the $l$-th bus is denoted by $P_{kl}(t)$, $Y_{kl}$ is the line admittance between the two buses, and $\phi_{kl}$ is the admittance phase angle \cite{kundur}.

The Center of Inertia (COI) is a commonly used concept for analyzing the overall dynamic behavior of a power system \cite{8523814}.The angular frequency at the COI is defined as the inertia-weighted average of individual bus frequencies:
\begin{align}\label{COI}
    \Delta\omega_0(t) = \frac{\sum_{k=1}^mH_k\omega_k(t)}{\sum_{k=1}^mH_k} - \omega^\star,
\end{align}
where $\Delta\omega_0(t)$ denotes the frequency deviation at the COI. It will be treated as a disturbance generated by the exosystem and as an exogenous input to the upcoming optimization problems. The framework to be developed is model-free with respect to (\ref{delta})–(\ref{COI}), leveraging the exosystem and, more specifically, frequency measurements.
%The difference between local and COI frequencies can be well estimated by an exponentially decaying sinusoid \cite{9162502}.

\begin{remark}
With the increasing deployment of Phasor Measurement Units (PMUs), synchrophasor data can now be sampled nearly 100 times faster than with traditional SCADA systems. Data from multiple PMUs are collected by a phasor data concentrator, which performs time alignment, noise filtering, and cleansing of corrupted data. In addition to voltage and current phasors, PMUs also measure frequency and its time derivative, enabling consistent aggregation over wide geographic areas and providing system operators with real-time access to $\Delta\omega_0(t)$ and $\Delta\dot\omega_0(t)$, which can be then broadcast to available response units \cite{6582699}.
\end{remark}

\section{Frequency-Varying Optimization}
\subsection{Frequency-Varying Equality Constraint}
Let $c_{agg}$ denote the aggregate contracted quantity of the ARU. The contract is submitted in advance and remunerated based on availability. In response to frequency events, each asset must adjust its power output relative to its operational baseline, and the delivered quantity of each asset is
\begin{align}\label{x}
    x_i(t) = P_i(t) - P_i(0),
\end{align}
where $P_i(t)$ is the power output of the $i$-th asset, and $P_i(0)$ represents its operational baseline submitted to NESO. For a BESS, negative values of $P_i(t)$ indicate charging, while positive values indicate discharging. The baseline $P_i(0)$ is typically scheduled by the Battery Management System (BMS) for purposes such as energy arbitrage and State-of-Charge (SoC) recovery \cite{8666790} and is not expected to change over very short timescales. The power output of assets are subject to state constraints, characterized by the set
\begin{align}\label{domain}
    \Omega_i = \{x_i(t)\mid P_i^\text{min}\leq P_i(0) + x_i(t)\leq P_i^\text{max}\},
\end{align}
where $P_i^\text{min}$ and $P_i^\text{max}$ are the minimum and maximum power limits for $P_i(t)$, respectively.

To collectively meet the contractual commitment, the delivered quantity should match the required quantity, which varies with $\Delta\omega_0(t)$ from 0\% to $\pm$100\% of the aggregate contracted quantity. This yields a frequency-varying equality constraint:
\begin{align}\label{equality constraint}
    \sum_{i=1}^n x_i(t) = h(\Delta\omega_0(t))\cdot c_{agg},
\end{align}
where $\sum_{i=1}^n x_i(t)$ is the overall response of the ARU, $h(\Delta\omega_0(t)):\mathbb R\to\mathbb R$ represents the delivery requirement curve as depicted in Fig. 1, and $h(\Delta\omega_0(t))\cdot c_{agg}$ is the required quantity of DR/DM/DC, calculated using frequency measurement. 

We illustrate the frequency-varying nature of (\ref{equality constraint}) by taking DM as an example. DM operates over a frequency deviation range of $\pm$0.2 Hz, with a dead-band of $\pm$0.015 Hz during which response is not triggered. Sublinear delivery is required for deviations between 0.015 Hz and 0.1 Hz, as well as between -0.1 Hz and -0.015 Hz. The knee points, occurring at $\pm$0.1 Hz, correspond to a response level of 5\% and -5\%, beyond which linear delivery begins. Linear delivery continues until the saturation points at $\pm$0.2 Hz, where full delivery ($\pm$100\%) is required.

\subsection{Problem Setup}
Each asset is supposed to have an cost function, denoted by $f_i(x_i(t), \Delta\omega_0(t))$, which may vary with frequency, though this is not necessarily the case. The optimal coordination of assets within the ARU is formulated as the following FVO problem. By replacing the delivery requirement curve, this formulation can be applied to DR, DM, or DC, triggering the specific one as needed:
\begin{optimization}{Problem: FVO-1}
\begin{align}
    &\underset{x(t)\in\mathbb R^n}{\text{min}}\sum_{i=1}^nf_i(x_i(t),\Delta\omega_0(t)),\notag\\
    &\text{s.t. }\sum_{i=1}^n x_i(t) = h(\Delta\omega_0(t))\cdot c_{agg},\notag\\
    &\quad x_i(t)\in\Omega_i\ \forall i=1,\dots,n,\notag
\end{align}
\text{with $\Delta\omega_0(t)$ being an exogenous input to the problem.}
\end{optimization}
The cost functions in power system applications are typically modeled as quadratic functions \cite{6980137,performance_optimization,9525453}, implying strong convexity. We assume strong duality holds (\textit{e.g.}, in convex problems satisfying Slater’s condition), under which a solution to the dual problem exists and the primal and dual problems have zero duality gap. This assumption is standard in the literature and fundamental for primal-dual methods \cite{6980137}.

\begin{assumption}\label{asp:convex}
For each $i\in\mathcal N$, the cost function $f_i(x_i,\Delta\omega_0)$ is twice continuously differentiable and strongly convex with respect to $x_i$, and continuously differentiable with respect to $\Delta\omega_0$.
\end{assumption}

\begin{assumption}\label{asp:slater}
The Slater’s condition holds for all $t\geq 0$, \textit{i.e.}, there exists at least one interior point satisfying (\ref{equality constraint}).
\end{assumption}

In the rest of this paper, we may omit ``$(t)$'' in derivations and analyses for notational simplicity. Let $f_{ix}(x_i,\Delta\omega_0):\Omega_i\times\mathbb R\to\mathbb R$ be the partial derivative of $f_i(x_i,\Delta\omega_0)$ with respect to $x_i$, $f_{ixx}(x_i,\Delta\omega_0)\in\Omega_i\times\mathbb R\to\mathbb R_{>0}$ and $f_{ix\omega}(x_i,\Delta\omega_0):\Omega_i\times\mathbb R\to\mathbb R$ be the partial derivatives of $f_{ix}(x_i,\Delta\omega_0)$ with respect to $x_i$ and $\Delta\omega_0$. Let $h_\omega:\mathbb{R} \to \mathbb{R}_{\leq 0}$ denote the derivative of $h(\Delta\omega_0)$, a piecewise constant function of $\Delta\omega_0$. 

%With phasor measurement units and related estimation techniques, $\Delta\omega_0$ is measurable, and $\Delta\dot\omega_0$ and $\Delta\ddot\omega_0$ can be computed from it \cite{9253523}.
We now present the existence and uniqueness of the optimal solution to Problem FVO-1 at time $t$, along with a characterization of its dynamics that forms the basis for our feedforward design. Due to the state constraints, the optimal trajectory is piecewise continuously differentiable and may fail to be differentiable at the switching instants---specifically, when transitioning between the interior $\mathrm{Int}(\Omega_i)$ and the boundary $\partial\Omega_i$ of the closed convex set $\Omega_i$. 
\begin{theorem}
Suppose that Assumptions 1--2 hold. Problem FVO-1 admits a unique optimal solution at each time $t$. Let $x^\star$ denote the optimal solution and $\lambda^\star$ the corresponding optimal Lagrange multiplier. Both $x^\star$ and $\lambda^\star$ are piecewise continuously differentiable, and for almost all $t\geq 0$,
\begin{align}
    \dot x_i^\star &= -\rho_i^\star\left[f_{ix\omega}(x_i^\star,\Delta\omega_0)\Delta\dot\omega_0 + \dot\lambda^\star\right],\label{optimal trajectory x}\\
    \dot\lambda^\star &= -\left(\sum_{i=1}^n\rho_i^\star\right)^{\dagger}\sum_{i=1}^n\rho_i^\star f_{ix\omega}(x_i^\star,\Delta\omega_0)\Delta\dot\omega_0 \notag\\
    &\quad -\left(\sum_{i=1}^n\rho_i^\star\right)^{\dagger}h_\omega(\Delta\omega_0)\Delta\dot\omega_0c_{agg},\label{optimal trajectory la}
\end{align}
where $\rho_i^\star$ denotes $\rho_i$ evaluated at $x_i = x_i^\star$, with $\rho_i:\Omega_i\times\mathbb R\to\mathbb R_{\geq0}$ given by
\begin{align}\label{rho}
    \rho_i = \sigma_if_{ixx}(x_i,\Delta\omega_0)^{-1},
\end{align}
and $\sigma_i\in\{0,1\}$ is a state-dependent switching signal defined as $\sigma_i=1$ if $x_i\in\mathrm{Int}(\Omega_i)$ and $\sigma_i=0$ if $x_i\in\partial\Omega_i$.
\end{theorem}

\begin{proof}
The proof is provided in Appendix B. We note that $f_{ixx}(x_i,\Delta\omega_0)^{-1}$ exists and is upper bounded due to the strong convexity of $f_i(x_i,\Delta\omega_0)$. The inverse of the sum $\sum_{i=1}^n \rho_i$ is upper bounded when $\sum_{i=1}^n \rho_i\neq 0$. When $\sum_{i=1}^n \rho_i = 0$, it follows that $\dot x_i^\star = 0$ and $\dot\lambda^\star = 0$. In the case $\sum_{i=1}^n \rho_i = 0$, we have $\dot x_i^\star = 0$ and $\dot\lambda^\star = 0$. Hence the pseudoinverse $(\cdot)^\dagger$ is used to ensure that (\ref{optimal trajectory x})–(\ref{optimal trajectory la}) are well defined.
\end{proof}

\subsection{Problem Reformulations}
In Problem FVO-1, (\ref{equality constraint}) entails instantaneous service delivery, which is generally infeasible in any optimization or control framework. Instead, NESO specifies a maximum delivery time $T_{del}^{\max}$, which in this sense represents an upper bound on the time for the state trajectory to converge to the optimal trajectory. Suppose $x_i(t)$ is adjusted via
\begin{align}\label{first-order}
    \dot x_i(t) = u_i(t),
\end{align}
where $u_i(t)$ is the optimization signal. Let $T$ denote the convergence time, and $T^{\max}$ its theoretical upper bound. For implementation purposes, we introduce the following TOT problem whose goal is to find and track the unknown optimal trajectory of Problem FVO-1:
\begin{optimization}{Problem: TOT-1}
For $i\in\mathcal N$, design $u_i(t)$ under (\ref{first-order}) such that
\begin{align}
    &u_i(t) \to^T \dot x_i^\star(t),\notag\\
    &x_i(t) \to^T x_i^\star(t),\ \ x_i(t)\in\Omega_i\notag\\
    &T \leq T^{\max} \leq T_{del}^{\max},\notag
\end{align}
where $x^\star(t)$ is the optimizer of Problem FVO-1 at time $t$.
\end{optimization}

\begin{remark}
    It is worth noting that (\ref{optimal trajectory x})--(\ref{optimal trajectory la}) have discontinuous right-hand sides. Accordingly, the solution of (\ref{first-order}) should be understood in the sense of Filippov solutions \cite{filippov2013differential}. This interprets $u_i(t)$ as a differential inclusion and its convergence to a set-valued map. Nevertheless, the common Lyapunov function approach and proofs analogous to those for smooth systems are still applicable, since the Lyapunov functions in this paper are continuous and monotonically decreasing at points of discontinuity.
\end{remark}

Recent auction results from NESO highlight the dominant role of BESSs in the new dynamic services \cite{NGESO_Auction}. Accordingly, we focus on BESSs as the representative assets in this paper. It is worth noting that the responses of assets are governed by inherent dynamics, which can be non-negligible even for BESSs. Their response times ranging from tens to hundreds of milliseconds may impair tracking performance if not properly addressed \cite{8424043}. With an algorithm that optimizes the coordination in real time, the closed-loop response (see Fig. 3) of the $i$-th BESS can be described as
\begin{align}
    \tau_i\dot x_i(t) &=  r_i(t)-x_i(t),\label{second-order-1}\\
    \dot r_i(t) &= u_i(t),\label{second-order-2}
\end{align}
where $\tau_i$ is a time constant determined by the equivalent impedance of the inverter-side filter and transformer; $x_i(t)\in\Omega_i$ is the delivered quantity, subject to state constraints; $r_i(t)\in\Omega_i$ is the control input, subject to input constraints; and $u_i(t)\in\mathbb R$ is the optimization signal for adjusting $r_i(t)$. The SoC dynamics are much slower and therefore omitted from the above.
\begin{figure}[!h]
	\centerline{\includegraphics[width=0.4\textwidth]{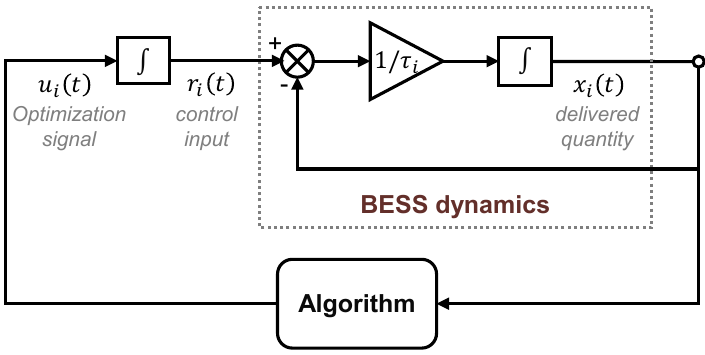}}
	\caption{Closed-loop response of BESS.}
\end{figure}
\begin{figure*}[h]
	\centerline{\includegraphics[width=0.7\textwidth]{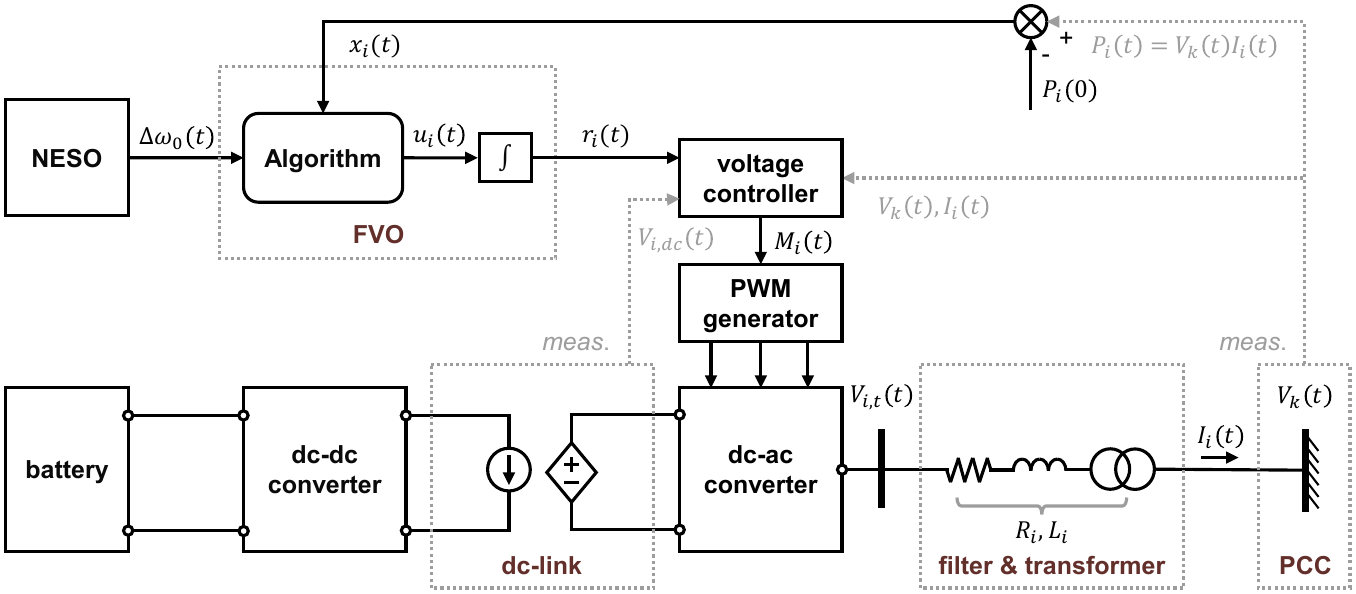}}
	\caption{Block diagram of BESS control system.}
	\label{power controller}
\end{figure*}

If $\tau_i$ is sufficiently small, $x_i(t) = r_i(t)$ holds approximately, and (\ref{second-order-1})--(\ref{second-order-2}) reduce to (\ref{first-order}). In practice, $\tau_i$ is not negligible, and the dynamics in (\ref{second-order-1}) must be explicitly taken into account. This makes TOT more challenging and motivates the developments in this paper. From (\ref{second-order-1})--(\ref{second-order-2}), it follows that $r_i(t)=\tau_i\dot x_i^\star(t) + x_i^\star(t)$ and $u_i(t) = \dot r_i(t)=\tau_i\ddot x_i^\star(t) + \dot x_i^\star(t)$ when $x_i(t) = x_i^\star(t)$. We arrive at the following TOT problem:
\begin{optimization}{Problem: TOT-2}
For $i\in\mathcal N$, design $u_i(t)$ under (\ref{second-order-1})--(\ref{second-order-2}) such that
\begin{align}
    &u_i(t) \to^T \tau_i\ddot x_i^\star(t) + \dot x_i^\star(t),\notag\\
    &r_i(t) \to^T \tau_i\dot x_i^\star(t) + x_i^\star(t),\ r_i(t)\in\Omega_i,\notag\\
    &x_i(t) \to^T x_i^\star(t),\ x_i(t)\in\Omega_i\notag\\
    &T \leq T^{\max} \leq T_{del}^{\max},\notag
\end{align}
where $x^\star(t)$ is the optimizer of Problem FVO-1 at time $t$.
\end{optimization}

\begin{remark}
As illustrated in Fig. 4, a BESS consists of battery modules, a dc-dc converter for adjusting voltage levels, and an inverter for grid interfacing. The circuit between the $i$-th BESS and its connected bus is described by $L_i\dot I_i(t) + R_i I_i(t) + V_k(t) = V_{i,t}(t)$ \cite{6397579}, where $I_i(t)$ is the output current, $V_{i,t}(t)$ is the terminal voltage magnitude, and $R_i$, $L_i$ denote the equivalent resistance and inductance of the inverter-side filter and transform, respectively \cite{9166575}. Multiplying both sides by $V_k(t)\cos(\phi_i)$ identifies the following expressions for (\ref{second-order-1}):
\begin{align}
    \tau_i &= L_i/R_i,\\
    r_i(t) &= [V_{i,t}(t)V_k(t)-V_k(t)^2]\cos(\phi_i)/R_i - P_i(0) \notag\\
    &\quad + \tau_i\dot V_k(t)I_i(t)\cos(\phi_i).\label{r}
\end{align}
In (\ref{r}), $V_{i,t}(t)$ is adjusted through Pulse Width Modulation (PWM), while $V_k(t)$, $I_i(t)$, and $V_{i,dc}(t)$ are locally measured. For actual implementation, $r_i(t)$ needs to be translated into an actuating signal---namely the modulation index $M_i(t)$ of the PWM generator \cite{9166575}:
\begin{align}
    M_i(t) = \sqrt{2}V_{i,t}(t)/(\sqrt{3}A_iV_{i,dc}(t)),\label{M}
\end{align}
where $A_i$ is a constant depending on the inverter topology, typically 0.5 or 1, and $V_{i,dc}(t)$ is the dc-link voltage. Based on (\ref{r})--(\ref{M}), $M_i$ can be computed from $r_i(t)$ and the local measurements. To decouple the control design for frequency and voltage regulation, it is common practice to approximate $V_k(t)$ by its nominal value \cite{6702462,7163587}.
\end{remark}

\section{Algorithmic Design}
\subsection{Algorithm for Problem TOT-1}
In this section, we present the algorithmic design for Problem TOT-1, followed by its convergence analysis. 

The Lagrangian function for Problem FVO-1 is
    \begin{align}\label{lagrangian}
        L &= \sum_{i=1}^n f_i(x_i,\Delta\omega_0) + \lambda\left(\sum_{i=1}^n x_i - h(\Delta\omega_0)c_{agg}\right),\ x_i\in\Omega_i,
    \end{align}
    whose partial derivatives can be obtained as
    \begin{align}
        \frac{\partial L}{\partial x_i} = f_{ix}(x_i,\Delta\omega_0)+\lambda,\quad 
        \frac{\partial L}{\partial\lambda} = \sum_{i=1}^n x_i - h(\Delta\omega_0)c_{agg}.\notag
    \end{align}
    At the optimal solution, the following hold for any $\kappa_x\in\mathbb R_{>0}$:
    \begin{align}
        0 &= x_i^\star - \mathcal P_{\Omega_i}(x_i^\star-\kappa_x\frac{\partial L}{\partial x_i}\mid_{x_i=x_i^\star,\lambda=\lambda^\star}),\label{optimality-1}\\
        0 &= \frac{\partial L}{\partial\lambda}\mid_{x=x^\star,\lambda=\lambda^\star}.\label{optimality-2}
    \end{align}
    These serve as necessary and sufficient conditions for the optimality of $x^\star$. In particular, (\ref{optimality-1}) stems from variational inequality theory and related optimization problems, and has been widely used in the design of projected dynamical systems. We refer interested readers to \cite{1288236,zeidler2013nonlinear} for background and details. 

    Inspired by (\ref{optimality-1}), which can be viewed as a sliding surface \cite{xu2008sliding} that the optimal trajectory stays on, we consider designing a feedforward driving term to drive the state to evolve along the optimal trajectory ``tangentially''. Meanwhile, a feedback driving term based on fixed-time control \cite{8322314} will be responsible for pushing the state onto the sliding surface ``perpendicularly'', as illustrated in Fig. 2. The algorithm for Problem TOT-1, which inherently supports dynamic environments such as the plug-in and plug-out of assets and other response units, is proposed as follows:
\begin{align}
    u_i &= -\underbrace{f_{ixx}(x_i,\Delta\omega_0)^{-1}(\gamma_{1}e_i^{1-\frac{p}{q}}+\gamma_{2}e_i^{1 + \frac{p}{q}}+\gamma_{3}\text{sign}(e_i))}_\text{FB driving term}\notag\\
    &\quad + \underbrace{\alpha_i\text{sig}(e_i)}_\text{FF driving term},\label{u single}\\
    e_i &= x_i - \mathcal P_{\Omega_i}(x_i - F_i),\label{e single}\\
    F_i &= \kappa_x\left[f_{ix}(x_i,\Delta\omega_0)+\lambda\right],\label{F single}\\
    \dot\lambda &= \kappa_\lambda\left(\sum_{i=1}^nx_i - h(\Delta\omega_0)c_{agg}\right) + \beta,\label{la single}
\end{align}
$x_i(0)\in\Omega_i\ \forall i\in\mathcal N$ and $\lambda(0)\in\mathbb R$. In the above, $e_i$ represents the error with respect to the stationarity condition in (\ref{optimality-1}) and is forced to zero within a fixed time through the feedback driving term, $F_i$ is a term related to gradient descent, $\gamma_{1},\gamma_{2},\gamma_{3}\in\mathbb R_{>0}$ are fixed-time control gains, $p,q\in\mathbb Z_{>0}$ are even/odd integers satisfying $p<q$, $\kappa_{x},\kappa_{\lambda}\in\mathbb R_{>0}$ are step-sizes, $\mathcal{P}_{\Omega_i}:\mathbb{R}\to\Omega_i$ denotes the projection operator, and $\alpha_i$, $\beta_i$ are the feedforward driving terms defined by (\ref{alpha})--(\ref{beta}).

\begin{lemma}\label{lemma feedforward single}
Suppose Assumptions 1--2 hold. Design the feedforward driving terms as
\begin{align}
    \alpha_i &= -\rho_i\left[f_{ix\omega}(x_i,\Delta\omega_0)\Delta\dot\omega_0 + \beta\right],\label{alpha}\\
    \beta &= -\left(\sum_{i=1}^n\rho_i\right)^{\dagger}\sum_{i=1}^n\rho_if_{ix\omega}(x_i,\Delta\omega_0)\Delta\dot\omega_0 \notag\\
    &\quad -\left(\sum_{i=1}^n\rho_i\right)^{\dagger}h_\omega(\Delta\omega_0)\Delta\dot\omega_0c_{agg},\label{beta}
\end{align}
where $\rho_i$ is defined as in Theorem 1. Then, $\alpha_i=\dot x_i^\star$ and $\beta=\dot\lambda^\star$ when $x=x^\star$. 
\end{lemma}

\begin{proof}
The proof follows straightforwardly by substituting $x=x^\star$ into (\ref{alpha})--(\ref{beta}) and comparing them with (\ref{optimal trajectory x})--(\ref{optimal trajectory la}).
\end{proof}

As shown in Fig. 5, the switching rule for $\sigma_i$ that determines $\rho_i$ is adjusted to synchronize with the activation of the projection operator, thereby avoiding unnecessary switching without altering the optimization results:
\begin{align}\label{sigma single}
\sigma_i =
\begin{cases}
    1 & \text{if } x_i-e_i\in\mathrm{Int}(\Omega_i),\\
    0 & \text{if } x_i-e_i\in\partial\Omega_i,
\end{cases}
\end{align} 
where $x_i-e_i=\mathcal P_{\Omega_i}(x_i-F_i)\in\Omega_i$. This adjustment is justified by fixed-time stability of $e_i$, which ensures that $e_i\to^T 0$ and $x_i-e_i\to^T x_i$. Thus, for $t\geq T$, the switching rule above is equivalent to the one in Theorem 1. Due to this equivalence, Lemma 3 remains valid and the optimality of the converged solution is preserved. A pseudocode for the proposed algorithms is also provided to aid comprehension.

\begin{figure}[h]	\centerline{\includegraphics[width=0.48\textwidth]{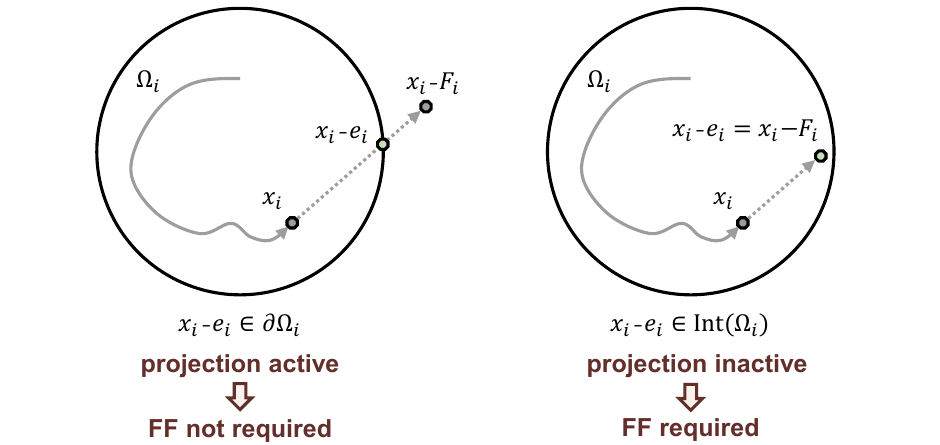}}
\caption{State-dependent switching rule based on projection.}
\end{figure}

\begin{algorithm}[t]
\caption{Algorithms for Problems TOT-1/TOT-2.}
\begin{algorithmic}[1]
\STATE Select control parameters according to $T_{del}^{\max}$;
\STATE Input $\Omega_i$, $x_i(0)\in\Omega_i$, and $c_{agg}$;
\STATE Set $\lambda(0)\in\mathbb R$ and $\sigma_i(0)\in\{0,1\}$;
\WHILE{there are no updates from BMS}
    \STATE Receive frequency measurements;
    \STATE Update required quantity $h_{\omega}(\Delta\omega_0)c_{agg}$;
    \STATE Update Lagrange multiplier $\lambda$ using (\ref{la single})/(\ref{la double});
    \STATE Update $e_i$ and $F_i$ using (\ref{e single})--(\ref{F single})/(\ref{e double})--(\ref{F double});
    \STATE Update state-dependent switching signal $\sigma_i$ using (\ref{sigma single})/(\ref{sigma double});
    \STATE Calculate feedforward driving terms using (\ref{alpha})--(\ref{beta})/(\ref{alpha''})--(\ref{beta'});
    \STATE Calculate and apply control signal $u_i$ according to (\ref{u single})/(\ref{u double});
    \STATE Observe delivered quantity $x_i$;
\ENDWHILE
\STATE Set $t\leftarrow 0$;
\STATE Go to 2;
\end{algorithmic}
\end{algorithm}

\begin{lemma}
For any $i\in\mathcal N$, if $x_i(0)\in\Omega_i$, then $x_i(t)\in\Omega_i$ and $(F_i-e_i)e_i\geq 0$ $\forall t\geq 0$.
\end{lemma}
\begin{proof}
The proof is provided in Appendix C.
\end{proof}

\begin{remark}
    Observe that $(F_i-e_i)e_i \cdot (F_i-e_i)\text{sign}(e_i) = (F_i-e_i)^2\vert e_i\vert \geq 0$, where $(F_i-e_i)e_i\geq 0$ as shown in Lemma 4. Then, it follows directly that $(F_i-e_i)\text{sign}(e_i)\geq 0$. Furthermore, we have $(F_i-e_i)e_i^{1-\frac{p}{q}} = (F_i-e_i)\vert e_i\vert^{1-\frac{p}{q}}\text{sign}(e_i) \geq 0$ and $(F_i-e_i)e_i^{1+\frac{p}{q}} = (F_i-e_i)\vert e_i\vert^{1+\frac{p}{q}}\text{sign}(e_i) \geq 0$. These properties will be used in the subsequent analysis and also hold when $F_i$ and $e_i$ are vectors rather than scalars; however, this requires a slightly different argument to show.
\end{remark}

\begin{theorem}
Suppose that Assumptions 1--2 hold for Problem TOT-1, which is solved using Algorithm 1 defined by (\ref{u single})--(\ref{sigma single}). Choose $\gamma_1,\gamma_2>0$ such that $\pi q/(2\kappa_xp\sqrt{\gamma_{1}\gamma_{2}}) \leq T_{del}^{\max}$. If $\gamma_{3}$ is sufficiently large to satisfy $\gamma_{3}\geq \Vert\dot\lambda+f_{ix\omega}\Delta\dot\omega_0\Vert_2$ $\forall i\in\mathcal N$ $\forall t\geq 0$, then $x_i=x_i^\star$ $\forall i\in\mathcal N$ $\forall t\geq T$, where
\begin{align}\label{Tmax}
    T \leq T^{\max} = \pi q /(2\kappa_xp\sqrt{\gamma_{1}\gamma_{2}})\leq T_{del}^{\max}.
\end{align}
\end{theorem}

\begin{proof}
Fixed any $i\in\mathcal N$. Consider the Lyapunov function
\begin{align}\label{lyapunov}
    W = (F_i-e_i)e_i + \frac{1}{2}e_i^2 \geq \frac{1}{2}e_i^2,
\end{align}
lower bounded as described in Lemma 4. We know $W=0$ if $e_i=0$, $W>0$ if $e_i\neq 0$, and $W\to+\infty$ if $e_i\to\infty$. For fixed-time stability of $e_i$, it remains to show that $\dot W$ is upper bounded as required for all $e_i\neq 0$.

To this end, we rewrite $W= -\text{min}_{y_i\in\Omega_i}W'$, where $W' = -F_i(x_i-y_i) + \frac{1}{2}(x_i-y_i)^2$ and the minimum is uniquely attained at $y_i=x_i-e_i$ \cite[Theorem 3.2]{fukushima1992equivalent}. Incorporate the fact that $\frac{\partial W}{\partial y_i}\mid_{y_i=x_i-e_i}=0$, the time derivative of $W$ can be obtained as
\begin{align}
    \dot W
    &= \underbrace{-\frac{\partial W'}{\partial x_i}\mid_{y_i=x_i-e_i} \dot x_i}_{\dot W_1} - \underbrace{\frac{\partial W'}{\partial\lambda}\mid_{y_i=x_i-e_i}\dot\lambda}_{\dot W_2}\notag\\
    &\quad \underbrace{-\frac{\partial W'}{\partial\Delta\omega_0}\mid_{y_i=x_i-e_i} \Delta\dot\omega_0}_{\dot W_3}.
\end{align}
We consider the case where $e_i\neq 0$ in the following derivations. The first term of $\dot W$ can be obtained as
\begin{align}\label{dV1i}
    \dot W_1
    &= \langle f_{ixx}(x_i,\Delta\omega_0)u_i,\kappa_xe_i\rangle  + \langle F_i - e_i,u_i\rangle \notag\\
    &= -\kappa_x(\gamma_{1}\vert e_i\vert^{2-\frac{p}{q}}+\gamma_{2}\vert e_i\vert^{2+\frac{p}{q}}+\gamma_{3}\vert e_i\vert)\notag\\
    &\quad - \gamma_{1}f_{ixx}(x_i,\Delta\omega_0)^{-1}(F_i-e_i)e_i^{1-\frac{p}{q}}\notag\\
    &\quad - \gamma_{2}f_{ixx}(x_i,\Delta\omega_0)^{-1}(F_i-e_i)e_i^{1+\frac{p}{q}}\notag\\
    &\quad - \gamma_{3}f_{ixx}(x_i,\Delta\omega_0)^{-1}(F_i-e_i)\text{sign}(e_i)\notag\\
    &\quad + (F_i-e_i)\alpha_i\text{sig}(e_i)\notag\\
    &\leq -\kappa_x(\gamma_{1}\vert e_i\vert^{2-\frac{p}{q}}+\gamma_{2}\vert e_i\vert^{2+\frac{p}{q}}+\gamma_{3}\vert e_i\vert),
\end{align}
where Remark 4 and $(F_i-e_i)\alpha_i\text{sig}(e_i)=0$ are used. The second and third terms of $\dot W$ are given by
\begin{align}
    \dot W_2 &= \kappa_x\kappa_\lambda \left(\sum_{i=1}^nx_i - h(\Delta\omega_0)c_{agg}\right)e_i + \kappa_x\beta e_i,\label{dV2i}\\
    \dot W_3 &= \kappa_xf_{ix\omega}(x_i,\Delta\omega_0)\Delta\dot\omega_0 e_i.\label{dV3i}
\end{align}
Provided that $\gamma_{3}\geq \Vert\dot\lambda+f_{ix\omega}\Delta\dot\omega_0\Vert_2$, we have
\begin{align}
    \dot W
    &\leq -\kappa_x(\gamma_{1}\vert e_i\vert^{2-\frac{p}{q}}+\gamma_{2}\vert e_i\vert^{2+\frac{p}{q}}).\label{dWi}
\end{align}
Recalling that $\sum_{i=1}^n e_i^2\leq 2W$, the summation of (\ref{dWi}) over $i$ gives rise to
\begin{align}
    \dot W \leq -(2^{1-\frac{p}{2q}}\kappa_x\gamma_{1})W^{1-\frac{p}{2q}}-(2^{1+\frac{p}{2q}}\kappa_x\gamma_{2})W^{1+\frac{p}{2q}}.
\end{align}
According to \cite[Theorem 3]{8322314}, and since $i$ was chosen arbitrarily,
\begin{align}
    e_i\to^T 0,\ x_i\to^T\alpha_i,\ \forall i\in\mathcal N,
\end{align}
where $T\leq T^{\max}\leq T_{del}^{\max}$. It is worth noting that $\alpha_i$ is locally Lipschitz continuous between consecutive switching instants for twice continuously differentiable convex $f_i(x_i,\Delta\omega_0)$.

Recall from (\ref{e single}) that $e_i = x_i-\mathcal P_{\Omega_i}(x_i-F_i)$. Hence, for all $i\in\mathcal N$ and $t\geq T$, there always exists $c_i \in \{y_i\in\mathbb R\mid x_i=\mathcal P_{\Omega_i}(x_i+y_i)\}$ such that
\begin{align}\label{stationarity}
    e_i = F_i + c_i = 0.
\end{align}
Furthermore, we have $\dot e_i = \dot F_i + \dot c_i = 0$ for all $t\geq T$, and $\dot c_i$ is well-defined except at the switching instants of $\sigma_i$ where non-differentiability arises. A candidate solution---subject to verification---is described by:
\begin{align}
    0 &= \sum_{i=1}^nx_i - h(\Delta\omega_0)c_{agg},\label{primal feasibility}\\
    0 &= f_{ixx}(x_i,\Delta\omega_0)\alpha_i + f_{ix\omega}(x_i,\Delta\omega_0)\Delta\dot\omega_0 + \beta + \dot c_i.\label{local feasibility}
\end{align} 
Note (\ref{stationarity})--(\ref{primal feasibility}) represent the stationarity and primal feasibility conditions, respectively, while (\ref{local feasibility}) holds at the optimal solution because $\alpha_i=\dot x_i^\star$ and $\beta=\dot\lambda^\star$. Thus, (\ref{primal feasibility})--(\ref{local feasibility}) are satisfied at $x=x^\star$. The existence of converged solutions has been established, and according to the Picard-Lindelöf theorem, $x^\star$ is the unique solution to (\ref{stationarity})--(\ref{local feasibility}) for all $t\geq T$ excluding the switching instants. On the other hand, optimality at the switching instants follows from the continuity of $x_i$ and $x_i^\star$. Consequently, $u_i\to^T \dot x_i^\star$ and $x_i\to^T x_i^\star$ with $T\leq T^{\max}\leq T_{del}^{\max}$. The proof is complete.
\end{proof}

\subsection{Algorithm for Problem TOT-2}
In this section, we present the algorithmic design for Problem TOT-2. Before that, we make an additional assumption:
\begin{assumption}\label{asp:convex}
For each $i\in\mathcal N$, the second derivative $f_{ixx}$ is constant, and $f_{ix\omega}(x_i,\Delta\omega_0)$ is continuously differentiable with respect to both $x_i$ and $\Delta\omega_0$.
\end{assumption}
Denote the partial derivatives of $f_{ixx}$ and $f_{ix\omega}$ by $f_{ixxx}$, $f_{ixx\omega}$, and $f_{ix\omega\omega}(x_i,\Delta\omega_0)$. Under Assumption 3, we have $f_{ixxx} = f_{ixx\omega} = 0$. Separately, $h_{\omega\omega}=0$ since $h_{\omega}(\Delta\omega_0)$ is a piecewise constant function. Thus, $x_i^\star$ and $\lambda^\star$ are piecewise twice continuously differentiable, and we arrive at the following corollary:

\begin{corollary}
Suppose Assumptions 1--3 hold. Design the feedforward driving terms as
\begin{align}
    \alpha_i'' &= -\tau_i\rho_i\left[f_{ix\omega}(x_i,\Delta\omega_0)\Delta\ddot\omega_0 + f_{ix\omega\omega}(x_i,\Delta\omega_0)\Delta\dot\omega_0^2 + \beta''\right]\notag\\
    &\quad + \alpha_i',\label{alpha''}\\
    \beta'' &= -\left(\sum_{i=1}^n\rho_i\right)^{\dagger}\sum_{i=1}^n\rho_if_{ix\omega}(x_i,\Delta\omega_0)\Delta\ddot\omega_0\notag\\
    &\quad -\left(\sum_{i=1}^n\rho_i\right)^{\dagger}\sum_{i=1}^n\rho_if_{ix\omega\omega}(x_i,\Delta\omega_0)\Delta\dot\omega_0^2\notag\\
    &\quad -\left(\sum_{i=1}^n\rho_i\right)^{\dagger}h_\omega(\Delta\omega_0)\Delta\ddot\omega_0c_{agg},\label{beta''}\\
    \alpha_i' &= -\rho_i\left[f_{ix\omega}(x_i,\Delta\omega_0)\Delta\dot\omega_0 + \beta'\right],\label{alpha'}\\
    \beta' &= -\left(\sum_{i=1}^n\rho_i\right)^{\dagger} \sum_{i=1}^n\rho_if_{ix\omega}(x_i,\Delta\omega_0)\Delta\dot\omega_0\notag\\
    &\quad -\left(\sum_{i=1}^n\rho_i\right)^{\dagger}h_\omega(\Delta\omega_0)\Delta\dot\omega_0c_{agg},\label{beta'}
\end{align}   
where $\rho_i$ is as defined in Theorem 1. Then, $\alpha_i''=\tau_i\ddot x_i^\star + \dot x_i^\star$, $\beta''=\ddot\lambda^\star$, $\alpha_i'=\dot x_i^\star$, and $\beta'=\dot\lambda^\star$ when $x=x^\star$. 
\end{corollary}

\vspace{3mm}

\begin{proof}
For Problem TOT-2, $\dot x_i^\star$ and $\dot\lambda^\star$ remain as given in (\ref{optimal trajectory x})--(\ref{optimal trajectory la}). Analytical expressions for $\ddot x_i^\star$ and $\ddot\lambda^\star$ can be readily obtained by taking the time derivatives of (\ref{optimal trajectory x})--(\ref{optimal trajectory la}). These expressions are omitted here for brevity, as they naturally appear in (\ref{alpha''})--(\ref{beta''}). Then, the proof can be completed by comparing $\alpha_i''$ with $\tau_i\ddot x_i^\star + \dot x_i^\star$, $\beta''$ with $\ddot\lambda^\star$, $\alpha_i'$ with $\dot x_i^\star$, and $\beta_i'$ with $\dot\lambda^\star$, after substituting $x = x^\star$.
\end{proof}

\vspace{-2mm}
Based on Corollary 1, the algorithm for Problem TOT-2 is proposed as follows:
\begin{align}
    u_i &= -\gamma_{1}e_i^{1-\frac{p}{q}}-\gamma_{2}e_i^{1+\frac{p}{q}}-\gamma_{3}\text{sign}(e_i) + \alpha_i''\text{sig}(e_i),\label{u double}\\
    e_i &= r_i - \mathcal P_{\Omega_i}(r_i-F_i),\label{e double}\\ 
    F_i &= \kappa_x\left[f_{ix}(x_i,\Delta\omega_0)+\lambda+(r_i-x_i-\tau_i\alpha_i')\right],\label{F double}\\
    \dot\lambda &= \kappa_\lambda\left(\sum_{i=1}^nx_i - h(\Delta\omega_0)c_{agg}\right)+\beta',\label{la double}
\end{align}
$r_i(0),x_i(0)\in\Omega_i$ and $\lambda(0)\in\mathbb R$, where $r_i-x_i-\tau_i\alpha_i'$ is introduced for convergence and vanishes at the optimal solution. The switching signal $\sigma_i$ is updated according to
\begin{align}\label{sigma double}
\sigma_i =
\begin{cases}
    1 & \text{if } r_i-e_i\in\mathrm{Int}(\Omega_i),\\
    0 & \text{if } r_i-e_i\in\partial\Omega_i,
\end{cases}
\end{align} 
which becomes equivalent to the one defined in Theorem 1 for all $t\geq T$, except when $r_i\in\partial\Omega_i$ and $x_i\in\text{Int}(\Omega_i)$, or $r_i\in\text{Int}(\Omega_i)$ and $x_i\in\partial\Omega_i$. the resultant sub-optimality stems from an unavoidable compromise to ensure feasibility of the input constraints but is transient due to (\ref{second-order-1}).

\begin{remark}
    Given $r_i(0)\in\Omega_i$, it follows that $r_i(t)\in\Omega_i$ for all $t\geq 0$, by an argument analogous to that used in the proof of Lemma 4. Thus, Lemma 4 and Remark 4 remain valid for (\ref{u double}), implying that $(F_i(t)-e_i(t))e_i(t)\geq 0$ and $(F_i(t)-e_i(t))u_i(t) \leq 0$ $\forall t\geq 0$. Furthermore, for $\delta\to 0^+$, $x_i(t+\delta) = x_i(t) + \delta\dot x_i(t) = x_i(t) + \delta/\tau_i(r_i(t)-x_i(t)) = (1-\delta/\tau_i)x_i(t) + \delta/\tau_ir_i(t)$, which, by convexity of $\Omega_i$, yields $x_i(t+\delta)\in\Omega_i$ if $x_i(t)\in\Omega_i$. With $x_i(0)\in\Omega_i$, one ensures $x_i(t)\in\Omega_i$ for all $t\geq 0$ as well.
\end{remark}

\begin{theorem}
Suppose that Assumptions 1--3 hold for Problem TOT-2, which is solved using Algorithm 2 defined by (\ref{alpha''})--(\ref{sigma double}). Choose $\gamma_1,\gamma_2>0$ such that $\pi q/(2\kappa_xp\sqrt{2\gamma_{1}\gamma_{2}}) \leq T_{del}^{\max}$. If $\gamma_{3}$ is sufficiently large to satisfy $\gamma_{3}\geq\Vert(f_{ixx}-1)\dot x_i + f_{ix\omega}\Delta\dot\omega_0+\dot\lambda-\tau_i\dot\alpha_i'\Vert_2$ $\forall i\in\mathcal N$ $\forall t\geq 0$, then $e_i\to^T 0$ and $x_i=x_i^\star$ $\forall i\in\mathcal N$ for any $t\geq T$ satisfying $r_i,x_i\in\text{Int}(\Omega_i)$ or $r_i,x_i\in\partial\Omega_i$, where $T$ is bounded by (\ref{Tmax}).
\end{theorem}

\begin{proof}
For any $i\in\mathcal N$. We adopt (\ref{lyapunov}) from the previous subsection for convergence analysis. By constructing $W'= -F_i(r_i-y_i) + \frac{1}{2}(r_i-y_i)^2$, we establish the relation $W = -\text{min}_{y_i\in\Omega_i}W'$, where the minimum is uniquely attained at $y_i=r_i-e_i$. The time derivative of $W$ can be obtained as 
\begin{align}
    \dot W &= \dot F_ie_i + (F_i-e_i)\dot r_i \leq \dot F_ie_i
\end{align}
due to $(F_i-e_i)u_i\leq 0$ and $\dot r_i = u_i$. Then, $\dot W$ can be shown to have an upper bound that is negative definite:
\begin{align}
     \dot W &\leq \langle f_{ixx}\dot x_i + f_{ix\omega}(x_i,\Delta\dot\omega_0)\Delta\dot\omega_0 +\dot\lambda,\kappa_xe_i\rangle \notag\\
     &\quad + \langle u_i-\dot x_i-\tau_i\dot\alpha_i',\kappa_xe_i\rangle  \notag\\
     &\leq -\kappa_x(\gamma_{1}\vert e_i\vert^{2-\frac{p}{q}}+\gamma_{2}\vert e_i\vert^{2+\frac{p}{q}}),
\end{align}
provided that $\gamma_{3}\geq\Vert(f_{ixx}-1)\dot x_i + f_{ix\omega}\Delta\dot\omega_0+\dot\lambda-\tau_i\dot\alpha_i'\Vert_2$. Note that $\alpha_i'$ is differentiable everywhere except at the switching instants. Again, we conclude that $e_i\to^T 0$ and $r_i-e_i\to^T r_i$. There always exists $c_i\in\{y_i\in\mathbb R\mid r_i=\mathcal P_{\Omega_i}(r_i+y_i)\}$ such that $e_i = F_i + c_i = 0$ and $\dot e_i = \dot F_i + \dot c_i = 0$ $\forall t\geq T$, where $r_i-x_i-\tau_i\alpha_i'= 0$ for $x_i=x_i^\star$. Following the similar arguments in the proof of Theorem 2, it can be shown that for any $t\geq T$ where either $r_i,x_i\in\text{Int}(\Omega_i)$ or $r_i,x_i\in\partial\Omega_i$, $x^\star$ represents the unique converged solution. The proof is thus complete.
\end{proof}

\begin{remark}
In the above, Algorithms 1--2 are organized in a centralized fashion for clarity and generality. However, it is important to emphasize that the proposed control framework is not restricted to a centralized implementation. The algorithms can be readily distributed by employing a methodology similar to our prior work \cite{arxiv}, as illustrated below using Algorithm 1 as an example. In essence, the computations of $\lambda$, $\alpha_i$, and $\beta$ require global information and therefore need to be carried out in a distributed manner. This can be achieved by applying a consensus protocol to obtain $\lambda_i$, a local estimate of $\lambda$, together with a distributed estimator to compute $\beta_i$, a local estimate of $\beta$. The value of $\alpha_i$ then follows directly from its relationship with $\beta_i$. This methodology applies to Algorithm 2 as well.
\end{remark}

\begin{remark}
    The following outlines connections and differences with some related works. \cite{7163587} establishes a foundational connection between frequency control and optimization. It, along with \cite{wang2024distributed}, can be viewed as a special case of FVO that includes a quadratic term of frequency deviation in the cost function and interprets frequency deviation as a dual variable associated with power balance. The closed loop formed by the algorithm and the physical power network converges to the steady-state optimal solution. However, tracking a trajectory that varies persistently with frequency while meeting the maximum delivery time requested by NESO, as in our work, demands complex algorithmic designs and convergence analyses different from \cite{7163587,wang2024distributed}. Compared to \cite{8673636}, which develops a feedback-based algorithm achieving exponential convergence with bounded tracking errors, our framework attains fixed-time convergence with vanishing tracking errors by incorporating both feedback and feedforward.
\end{remark}

\section{Case Studies}
\subsection{Simulation Setup}
Firstly, we verify the effectiveness of the proposed framework through case studies in MATLAB/SIMULINK on the IEEE-14 bus system, as shown in Fig. 6. The control interval is set to 1 ms, and the communication frequency is set to 1 kHz accordingly. Subsequently, the scalability of the framework is evaluated on the IEEE 39-bus system. Both systems are equipped with AGC, provided by synchronous generators, and constant bus voltage magnitudes are assumed throughout the simulations. System parameters are available in the IEEE datasheets, and the key parameters for optimization/control are listed in Table I. Based on the chosen control parameters, the upper bound on $T$ is calculated as $T^{\max}=0.785$ s. 
\begin{figure}[h]
\centerline{\includegraphics[width=0.45\textwidth]{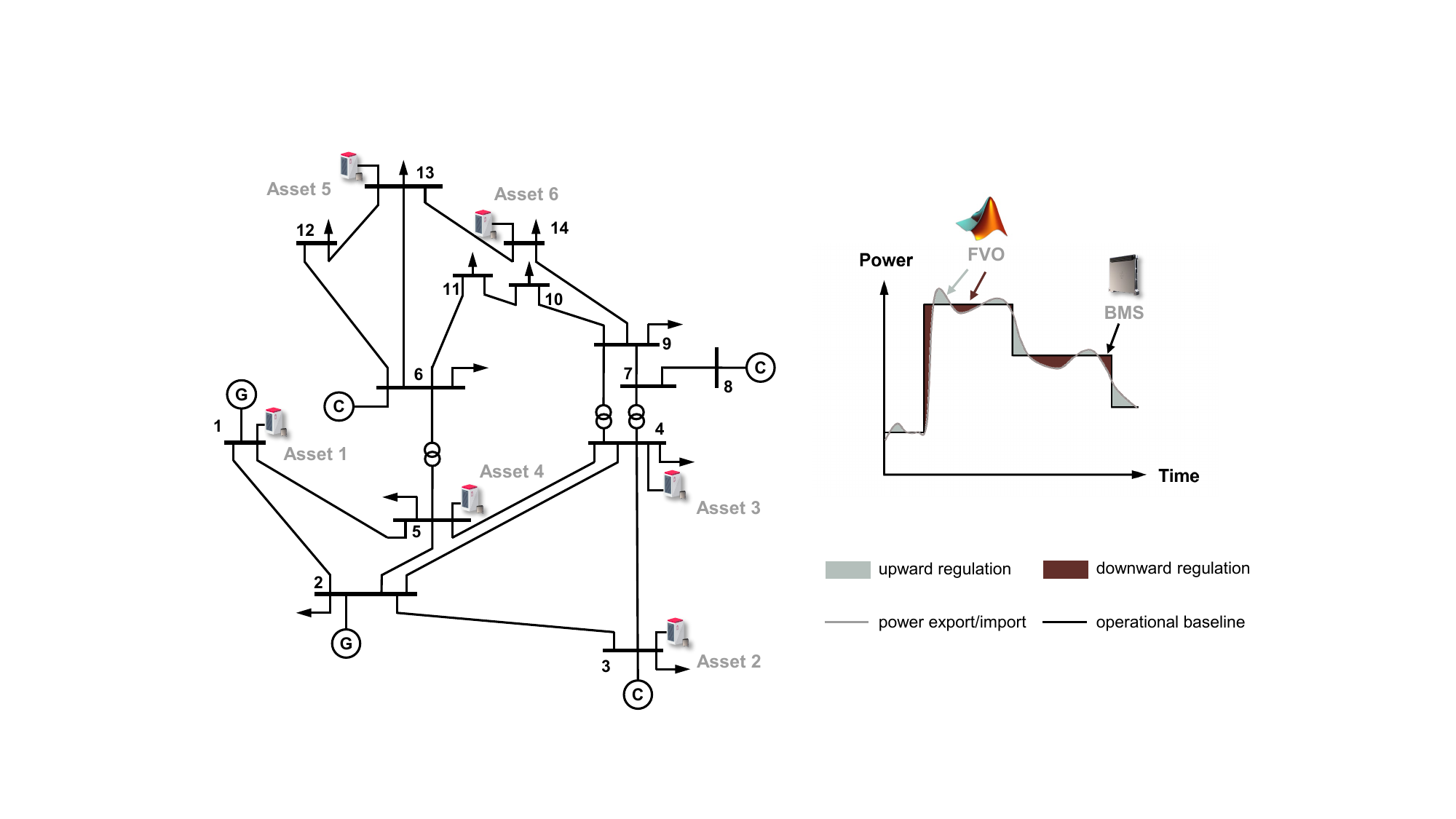}}
	\caption{The IEEE 14-bus system with an ARU comprising 6 BESSs.}
\end{figure}

The ARU receives an availability payment at the end of each settlement period (4 hours) based on its contract and availability. This payment can be treated as a constant as long as (\ref{equality constraint}) is well met, and is therefore omitted from the cost function. The cost function is modeled as a quadratic function that captures the cost of deviation from the operational baseline: $f_i(x_i(t)) = a_ix_i(t)^2 + b_ix_i(t)$, where $a_i\in\mathbb R_{>0}$ is a cost coefficient on deviating from the operational baseline, and $b_i\in\mathbb R_{>0}$ is a cost coefficient related to electricity import/export. The value of $a_i$ reflects the distance between the asset and the COI, such that assets closer to the COI contribute more actively. These cost functions can be replaced with frequency-varying ones if needs be. While the operational baselines are arbitrarily chosen in our simulation, in practice they should be optimized by the BMS, which operates at a higher hierarchical level.
\begin{table}[h]
    \centering
    \caption{Key parameters for simulation setup.}
    \label{table}
    \begin{tabular}
        {|p{52pt}|p{110pt}|}
        \hline
        \textbf{Parameter} & \textbf{Value} \\
        \hline
        $p$, $q$ & 2, 3 \\
        $\kappa_x$, $\kappa_\lambda$ & 1, 20 \\
        $\gamma_{1}$, $\gamma_{2}$, $\gamma_{3}$ & 3, 3, 200 \\
        $c_{agg}$ & 50 MW\\
        $a_i$ & 2.0, 3.2, 3.0, 2.4, 4.0, 5.0 \\
        $b_i$ & 1.0, 1.0, 1.0, 1.0, 1.0, 1.0 \\
        $\tau_i$ & 50, 160, 120, 200, 80, 150 ms \\
        $P_i^{\max}$ & 8.8, 7.7, 9.3, 17.3, 15.0, 8.0 MW \\
        $P_i(0)$ & 5.3, 1.7, -2.0, -4.3, -2.7, 6.7 MW\\
        \hline
    \end{tabular}
\end{table}

\subsection{Case Study: TOT-1}
We begin by demonstrating the effectiveness of Algorithm 1 for DC. This requires first specifying the delivery requirement curve, including $h(\Delta\omega_0)$ and $h_\omega(\Delta\omega_0)$. Since DC is a post-fault service designed to address large frequency deviations, a step increase in the net load of 1 pu is introduced at bus 2 at $t=20$ s to simulate a loss of renewable generation. As shown in Fig. 7(a), the grid frequency experiences a sudden drop but is gradually restored to the nominal value by AGC, where the black dotted lines represent the dead-band of service. Fig. 7(b) illustrates the quantity delivered by each asset that evolves dynamically due to frequency dynamics. In Figs. 7(c)--(d), we observe that the global cost function is minimized and the frequency response is delivered as contracted, where $\sum_{i=1}^n f_i(x_i^\star(t),\Delta\omega_0(t))$ is computed using the Gurobi solver. Algorithm 1 converges at approximately $t=20.2$ s, \textit{i.e.}, $T\approx0.2\leq T^{\max}=0.785\leq T_{del}^{\max} = 1$ s. When certain assets reach their state constraints, as indicated by the flattening of specific curves in Fig. 7(b), TOT remains effective. 
Next, we examine Algorithm 1 for DM under continuous net load fluctuations. Bounded arbitrary signals are generated to simulate the volatile periods targeted by DM. As shown in Fig. 8(a), stochastic fluctuations in the grid frequency are observed. Despite this, DM exhibits a relatively low utilization level in the sublinear delivery range and actively operates in the linear delivery range, as shown by the ripples in Fig. 8(b). According to Figs. 8(c)--(d), $T\approx0.3\leq T^{\max}=0.785\leq T_{del}^{\max} = 1$ s, and the optimality and feasibility of the state trajectory are well maintained irrespective of continuous net load fluctuations.
\begin{figure}[t] 
    \label{result_a1}
    \centering
    \begin{subfigure}{0.21\textwidth}
        \centering
        \includegraphics[width=\textwidth]{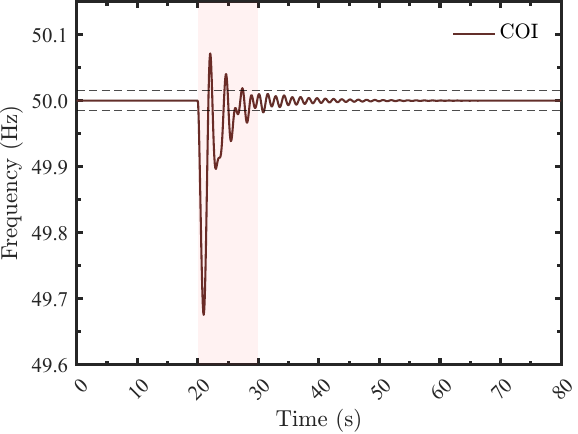}
        \caption{}
    \end{subfigure}
    \begin{subfigure}{0.2\textwidth}
        \centering
        \includegraphics[width=\textwidth]{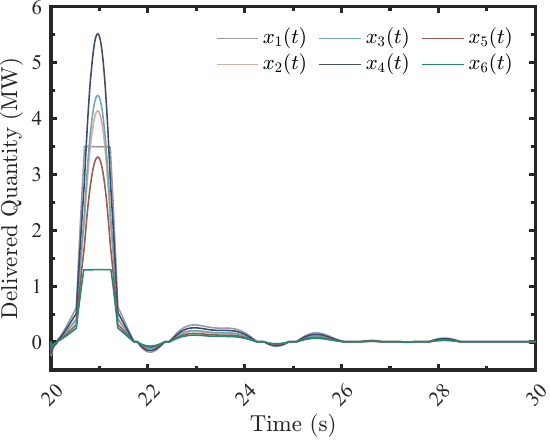}
        \caption{}
    \end{subfigure} 
    
    \begin{subfigure}{0.2\textwidth}
        \centering
        \includegraphics[width=\textwidth]{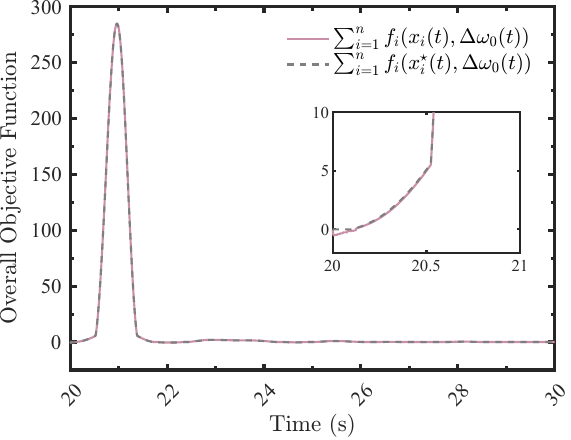}
        \caption{}
    \end{subfigure}
    \begin{subfigure}{0.2\textwidth}
        \centering
        \includegraphics[width=\textwidth]{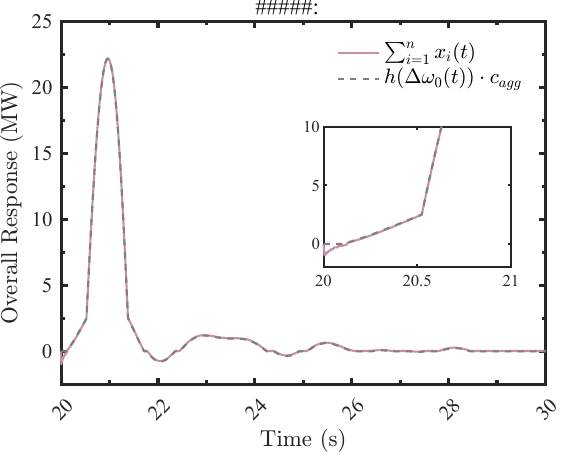}
        \caption{}
    \end{subfigure}
    \caption{Numerical results of Algorithm 1 for DC, solving Problem TOT-1 under step increase of net load: (a) COI frequency; (b) delivered quantity; (c) optimality; (d) feasibility.}
\end{figure}
\begin{figure}[t]
    \centering
    \begin{subfigure}{0.202\textwidth}
        \centering
        \includegraphics[width=\textwidth]{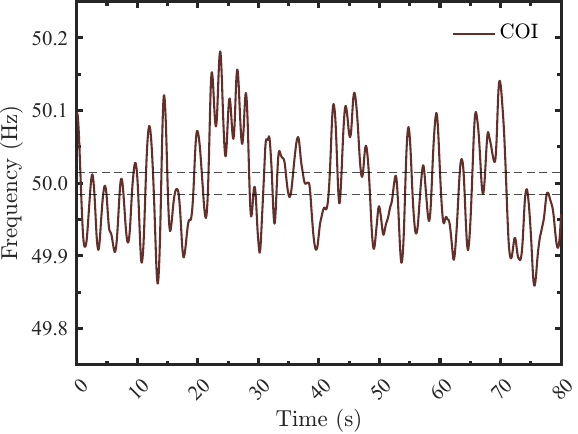}
        \caption{}
    \end{subfigure}
    \begin{subfigure}{0.2\textwidth}
        \centering
        \includegraphics[width=\textwidth]{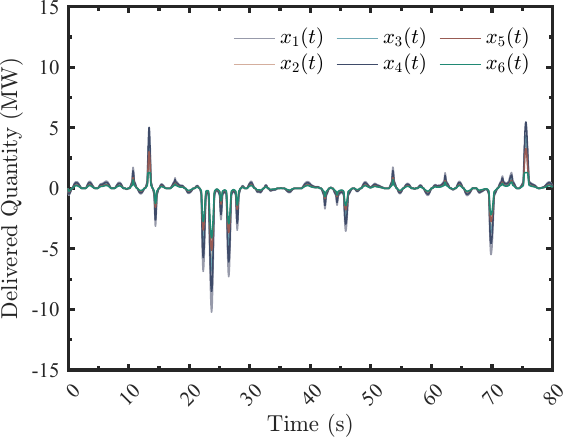}
        \caption{}
    \end{subfigure} 
    
    \begin{subfigure}{0.2\textwidth}
        \centering
        \includegraphics[width=\textwidth]{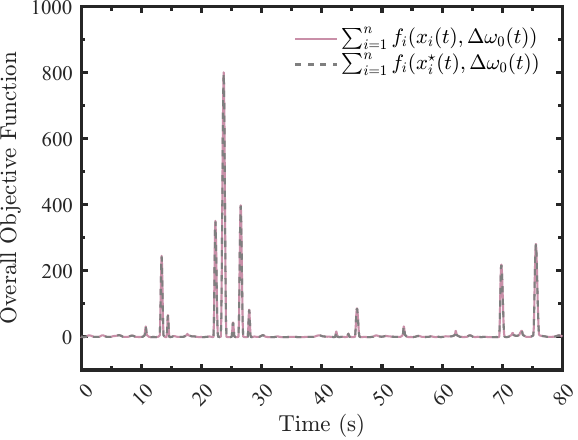}
        \caption{}
    \end{subfigure}
    \begin{subfigure}{0.2\textwidth}
        \centering
        \includegraphics[width=\textwidth]{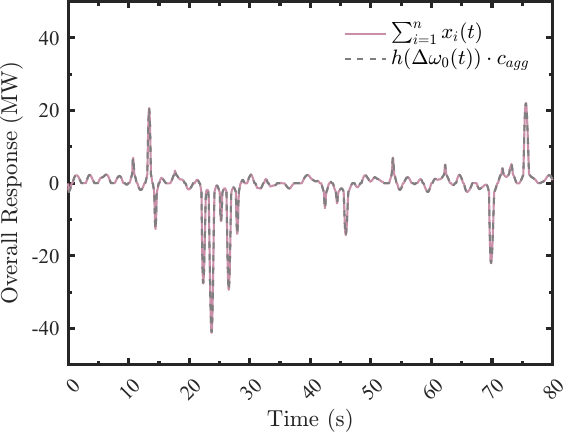}
        \caption{}
    \end{subfigure}
    \caption{Numerical results of Algorithm 1 for DM, solving Problem TOT-1 under continuous net load fluctuations: (a) COI frequency; (b) delivered quantity; (c) optimality; (d) feasibility.}
\end{figure}

\subsection{Case Study: TOT-2}
As previously discussed, Problem TOT-1 represents a relatively idealized scenario compared to Problem TOT-2. We now verify the effectiveness of Algorithm 2 in addressing Problem TOT-2, which accounts for the non-negligible dynamics associated with the inverter-side filter and transformer. 
Fig. 9 presents the results of Algorithm 2 for DR, solving Problem TOT-2 under continuous net load fluctuations. It is observed that $T\approx 0.3\leq T^{\max}=0.785\leq T_{del}^{\max} = 10$ s. The discrepancy between $r_i(t)$ and $x_i(t)$, even minor, can have a significant impact on tracking. This is demonstrated in Fig. 10, where Algorithms 1--2 are evaluated under the same frequency event triggering DC. While Algorithm 1 exhibits degraded tracking performance with noticeable oscillations, Algorithm 2 converges in approximately 0.4 s, followed by effective TOT.
\begin{figure}[t]
    \centering
    \begin{subfigure}{0.205\textwidth}
        \centering
        \includegraphics[width=\textwidth]{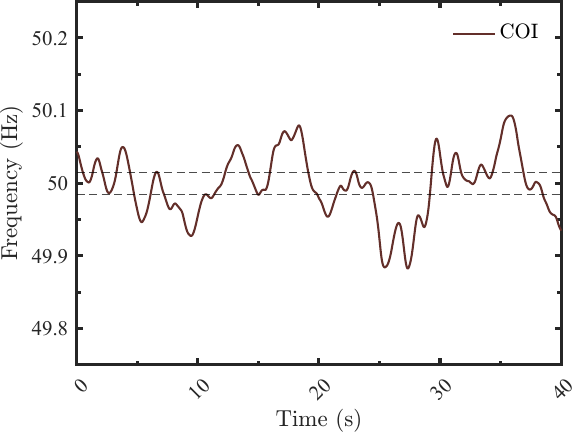}
        \caption{}
    \end{subfigure}
    \begin{subfigure}{0.2\textwidth}
        \centering
        \includegraphics[width=\textwidth]{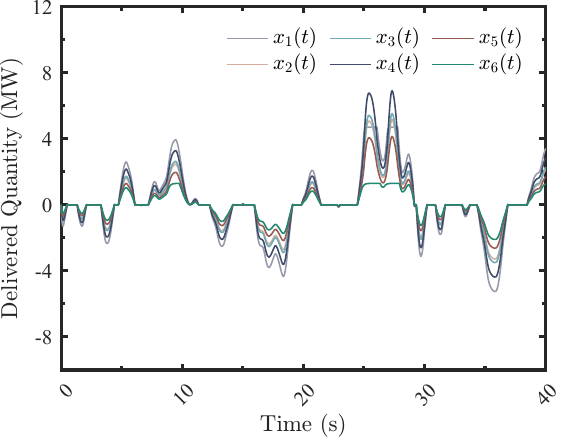}
        \caption{}
    \end{subfigure} 

    \begin{subfigure}{0.2\textwidth}
        \centering
        \includegraphics[width=\textwidth]{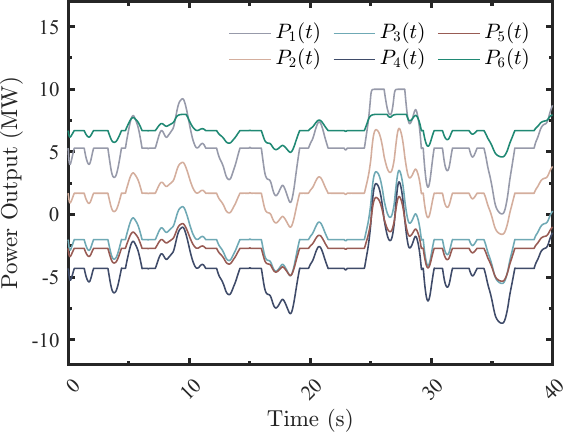}
        \caption{}
    \end{subfigure}
    \begin{subfigure}{0.2\textwidth}
        \centering
        \includegraphics[width=\textwidth]{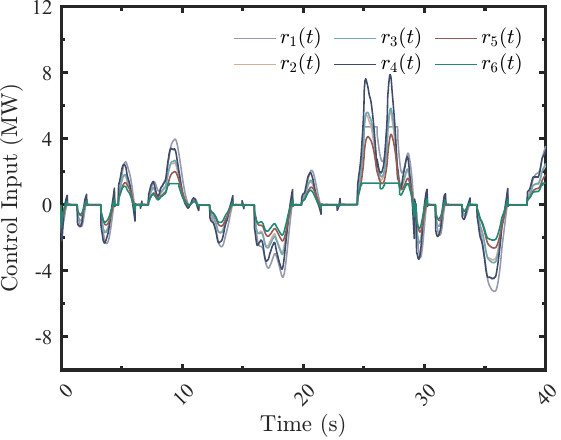}
        \caption{}
    \end{subfigure} 
    
    \begin{subfigure}{0.2\textwidth}
        \centering
        \includegraphics[width=\textwidth]{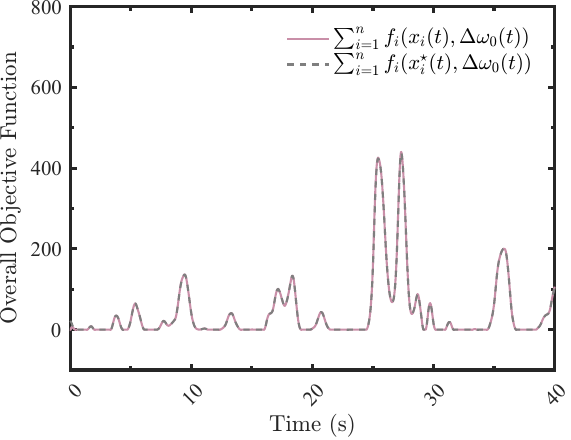}
        \caption{}
    \end{subfigure}
    \begin{subfigure}{0.2\textwidth}
        \centering
        \includegraphics[width=\textwidth]{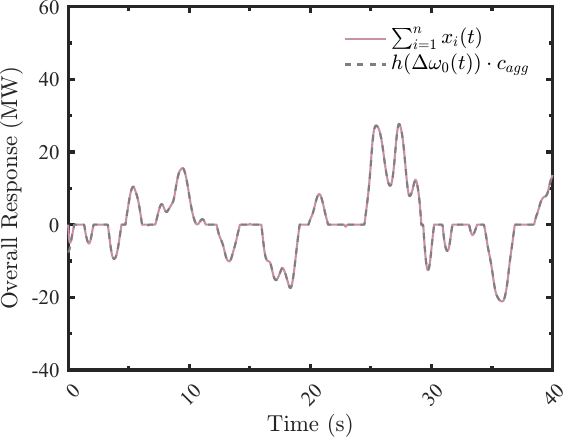}
        \caption{}
    \end{subfigure}
    \caption{Numerical results of Algorithm 2 for DR, solving Problem TOT-2 under continuous net load fluctuations: (a) COI frequency; (b) delivered quantity; (c) control input; (d) power output; (e) optimality; (f) feasibility.}
\end{figure}
\begin{figure}[h]
    \centering
    \begin{subfigure}{0.4\textwidth}
        \centering
        \includegraphics[width=\textwidth]{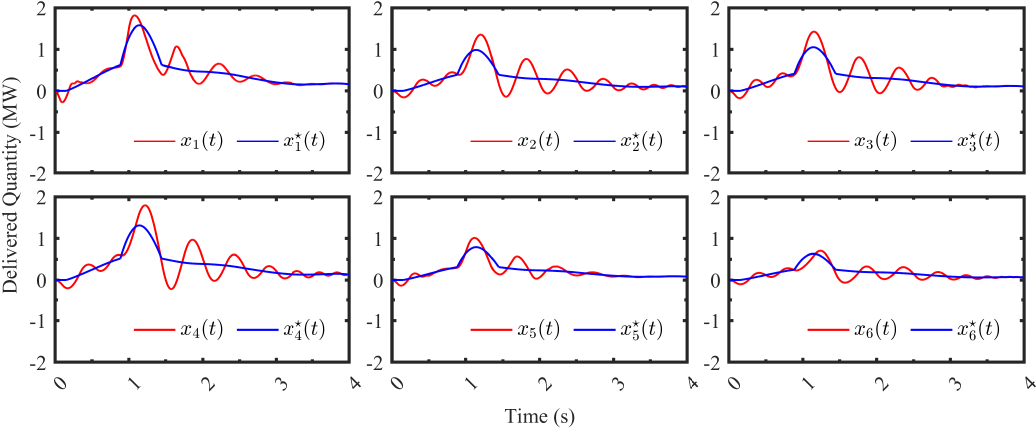}
        \caption{}
    \end{subfigure}
    \begin{subfigure}{0.4\textwidth}
        \centering
        \includegraphics[width=\textwidth]{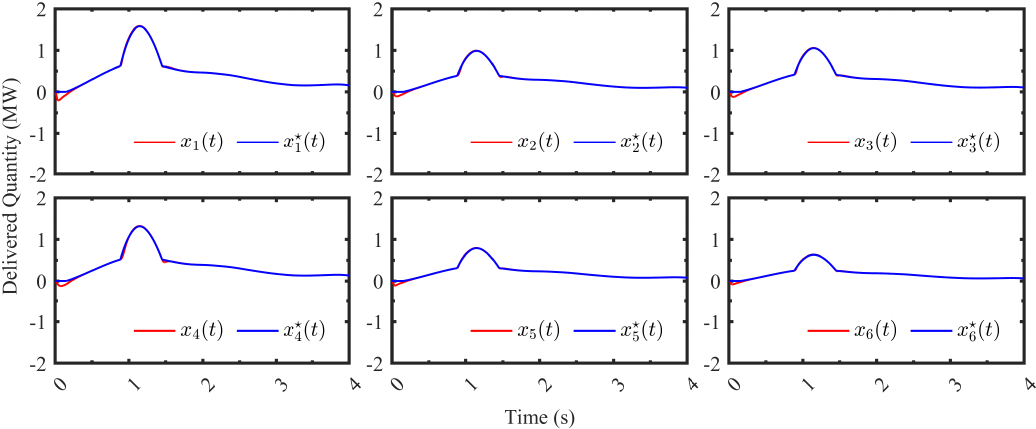}
        \caption{}
    \end{subfigure}
    \caption{Tracking performance of Algorithms 1--2 when considering BESS dynamics: (a) Algorithm 1; (b) Algorithm 2.}
\end{figure}

In what follows, we conduct scalability tests on the IEEE 39-bus system, which is partitioned into three geographical areas, each corresponding to an ARU comprising 30 BESSs located across the area, as shown in Fig. 11. Each ARU is contracted to deliver a total of 30 MW as part of the DR service. The cost parameters and time constants of the BESSs are randomly generated but kept identical across the three ARUs for comparison. The system is subject to continuous net load fluctuations, and different control strategies are deployed in each ARU: Algorithm 2 in ARU-1, Algorithm 1 in ARU-2, and a benchmark primal-dual projected gradient algorithm \cite{initialization_free} in ARU-3. the resultant frequency response mismatch $\sum_{i=1}^n x_i - h(\Delta\omega_0)c_{agg}$ are shown in Fig. 12. As illustrated, Algorithm 2 achieves the closest to ideal tracking performance, with the frequency response mismatch remaining nearly zero. Algorithm 1 exhibits larger oscillations compared to Algorithm 2, though still within an acceptable range. In contrast, the benchmark algorithm fails to achieve satisfactory optimization results in view of the amplitude of oscillations. This reveals the importance of fixed-time convergence and feedforward driving terms for TOT.
\begin{figure}[b]
\centerline{\includegraphics[width=0.4\textwidth]{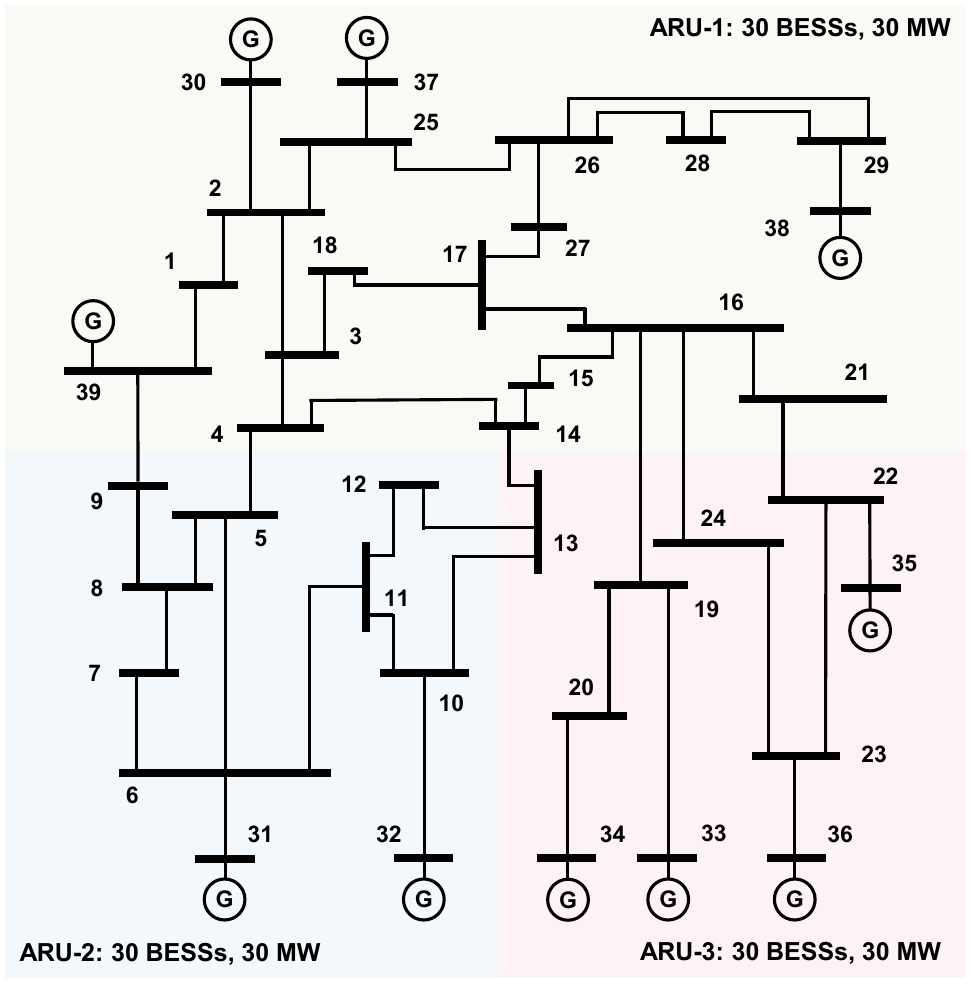}}
	\caption{The IEEE 39-bus system with three ARUs, each comprising 30 BESSs across the colored area.}
\end{figure}
\begin{figure}[htbp]
    \centering
    \begin{subfigure}{0.4\textwidth}
        \centering
        \includegraphics[width=\textwidth]{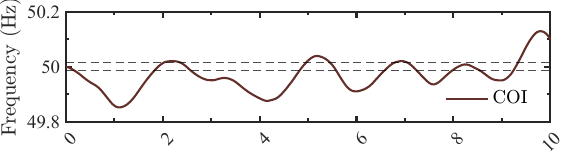}
    \end{subfigure}
    \begin{subfigure}{0.4\textwidth}
        \centering
        \includegraphics[width=\textwidth]{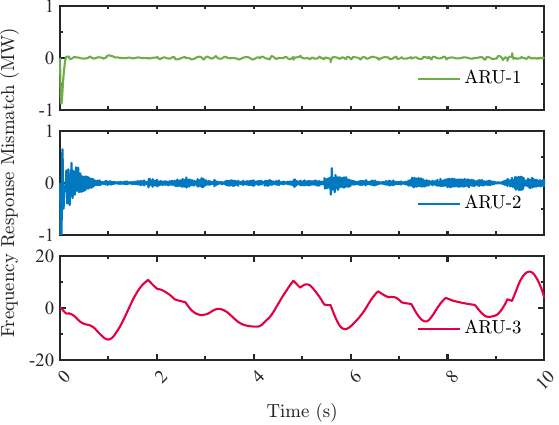}
    \end{subfigure}
    \caption{Frequency response mismatches in the three ARUs coordinated by Algorithm 2, Algorithm 1, and the benchmark algorithm, respectively.}
\end{figure}

Furthermore, we evaluate the average computational time of Algorithm 1 and Algorithm 2 by varying the number of assets, under centralized and distributed implementations. In a centralized implementation, an aggregator gathers all information and performs the computations. In a distributed implementation, each asset performs local computation based on its own information and information exchanged through the communication network. The communication topology can be very sparse, as long as a path connecting any two assets exists. Algorithms 1--2 are distributed according to Remark 6, adopting an undirected communication topology where each asset communicates with two other assets. The average computational time per control interval, evaluated under different test configurations, is shown in Table II, using an Intel i5-1135G7 processor (2.40 GHz, 4 cores, 8 threads, 16 GB RAM). The proposed control framework exhibits comparable computational time for both algorithms, with all computations completed easily within each control interval. While the average computational time for centralized implementation scales linearly with the number of assets, the distributed implementation maintains a low computational time that increases sublinearly. Overall, the framework is computationally efficient in both centralized and distributed implementations, capable of coordinating a large number of assets within an ARU.

\begin{table}[htbp]
\centering
\caption{Average computational time per control interval under different configurations.}
\begin{tabular}{|cc|cc|cc|}
\hline
\multicolumn{2}{|c|}{\multirow{2}{*}{Configuration}} & \multicolumn{2}{c|}{Algorithm 1}          & \multicolumn{2}{c|}{Algorithm 2}          \\ \cline{3-6} 
\multicolumn{2}{|c|}{}                               & \multicolumn{1}{c|}{Cent.} & Dist.  & \multicolumn{1}{c|}{Cent.} & Dist.  \\ \hline
\multicolumn{1}{|c|}{\multirow{3}{*}{Num.}}   & 30    & \multicolumn{1}{c|}{9 us}  & 0.7 us & \multicolumn{1}{c|}{10 us} & 0.8 us \\ \cline{2-6} 
\multicolumn{1}{|c|}{}                       & 60    & \multicolumn{1}{c|}{18 us} & 0.8 us & \multicolumn{1}{c|}{20 us} & 1.0 us \\ \cline{2-6} 
\multicolumn{1}{|c|}{}                       & 120   & \multicolumn{1}{c|}{36 us} & 1.0 us & \multicolumn{1}{c|}{40 us} & 1.2 us  \\ \hline
\end{tabular}
\end{table}

\section{Conclusions}
This paper has bridged frequency response services with FVO, demonstrating that the optimal coordination of an ARU constitutes an FVO problem, where the equality constraint varies with grid frequency and the optimal trajectory evolves dynamically. To facilitate online optimization, we have reformulated it into two TOT problems, accounting for asset dynamics. Algorithms incorporating feedback and feedforward terms have been proposed for the TOT problems, therefore solving the FVO problem online within the maximum delivery time while respecting the state/input constraints. The proposed framework has been demonstrated to be effective and scalable through case studies on the IEEE 14- and 39-bus systems. In terms of future directions, we are interested in investigating service stacking and distribution-level ARUs, which are becoming increasingly prevalent. Future work may also explore simultaneous voltage regulation and the use of hybrid energy storage for fast frequency response.

\section{Appendix}
\subsection{Proof of Lemma 1}
    1) Since $g(\epsilon,x)\to\mathbb I(x)=0$ as $\epsilon\to+\infty$ point-wise, it is straightforward to conclude $\lim_{\epsilon\to+\infty}g_{xx}(\epsilon,x)=0$ for all $x\in\mathrm{Int}(\Omega)$. 2) Fix any $x_0\in\partial\Omega$. Suppose, for contradiction, that $g_{xx}(\epsilon,x_0)$ remains bounded as $\epsilon\to+\infty$. Then, there exist finite $M$ and $\epsilon_0$ such that
    \begin{align}\notag
        \sup_{\epsilon\geq\epsilon_0}g_{xx}(\epsilon,x_0)\leq M.
    \end{align}
    Choose $\delta\in\mathbb R$ such that $x_0-\delta\in\mathrm{Int}(\Omega)$ and $x_0+\delta\notin\Omega$. By convexity, $g_x(\epsilon,x_0)\delta \leq g(\epsilon,x_0)-g(\epsilon,x_0-\delta)$, from which we know $\lim_{\epsilon\to+\infty}g_x(\epsilon,x_0)\delta\leq 0 - 0 = 0$. By applying Taylor’s theorem with Lagrange remainder, we have
    \begin{align}\notag
        g(\epsilon,x_0+\delta) = g(\epsilon,x_0)+ g_x(\epsilon,x_0)\delta +\tfrac12 g_{xx}(\epsilon,\xi)\delta^2
    \end{align}
    for some $\xi\in(0,\delta)$ and $g(\epsilon,x)$ that is twice continuously differentiable with respect to $x$. As a result,
    \begin{align}\notag
        \lim_{\epsilon\to+\infty}g(\epsilon,x_0+\delta)\leq\tfrac12 M\delta^2<+\infty,
    \end{align}
    which contradicts the definition that $\lim_{\epsilon\to+\infty}g(\epsilon,x_0+\delta)=\mathbb I(x_0+\delta)=+\infty$. Hence $\lim_{\epsilon\to+\infty}g_{xx}(\epsilon,x)=+\infty$ for all $x\in\partial\Omega$.
    The proof is complete.

\subsection{Proof of Theorem 1}
To facilitate our analysis, define $f_0(x,\Delta\omega_0) = \sum_{i=1}^n f_i(x_i,\Delta\omega_0)$. The feasible region for Problem FVO-1 is given by $\mathcal S = \{x\in\prod_{i=1}^n\Omega_i \mid 1_n^\top x = h(\Delta\omega_0)c_{agg}\}$, which is nonempty due to Slater's condition.
    
    \textbf{Existence:} 
    Since $f_0(x,\Delta\omega_0)$ is strongly convex in $x$, it follows that $f_0(x,\Delta\omega_0) \to \infty$ as $\Vert x\Vert \to \infty$. Let $x_0 \in \mathcal S$ be an arbitrary feasible point. The sublevel set $\mathcal L = \{x\in\mathcal S \mid f_0(x,\Delta\omega_0) \leq f_0(x_0,\Delta\omega_0)\}$ is closed (as the intersection of closed sets: affine equality and convex inequality constraints induce closed feasible regions) and bounded due to coercivity. By the Weierstrass theorem, a continuous function attains its minimum over a compact set. Hence, there exists $x^\star\in\mathcal L\subseteq \mathcal S$ such that
    \begin{align}\notag
        f_0(x^\star,\Delta\omega_0) = \inf_{x\in\mathcal S} f_0(x,\Delta\omega_0).
    \end{align}
    
    \textbf{Uniqueness:}  
    Suppose, for the sake of contradiction, two distinct optimal solutions $x_a, x_b \in \mathcal S$ exist with $f_0(x_a,\Delta\omega_0) = f_0(x_b,\Delta\omega_0) = f_0^\star$. Let $x_c = \frac{x_a + x_b}{2}$. By convexity of $\mathcal S$, we know $x_c \in \mathcal S$. Since $f_0$ is strongly convex, there exists $\mu\in\mathbb R_{>0}$ such that $f_0(x_c,\Delta\omega_0)\leq\frac{1}{2}f_0(x_a,\Delta\omega_0) + \frac{1}{2}f_0(x_b,\Delta\omega_0) - \frac{\mu}{8}\Vert x_a - x_b\Vert^2$. Substituting $f_0(x_a,\Delta\omega_0) = f_0(x_b,\Delta\omega_0) = f_0^\star$ yields 
    \begin{align}\notag
        f_0(x_c,\Delta\omega_0) = f_0^\star - \frac{\mu}{8}\Vert x_a - x_b\Vert^2 < f_0^\star,
    \end{align}
    which contradicts the optimality of $f_0^\star$. Therefore, the optimal solution must be unique.
    
    \textbf{Dynamics:} 
    By defining the indicator function $\mathbb I(x_i)=0$ if $x_i\in\Omega_i$ and $\mathbb I(x_i)=+\infty$ otherwise, we have
    \begin{align}
        x^\star &= \text{argmin}_{x\in\Omega}\sum_{i=1}^nf_i(x_i,\Delta\omega_0) \text{ s.t. (8)}\notag\\
        &= \text{argmin}_{x\in\mathbb R^n}\sum_{i=1}^n\left[f_i(x_i,\Delta\omega_0)+\mathbb I(x_i)\right] \text{ s.t. (8)}.
    \end{align}
    On the other hand, define the regularized cost function $\tilde f_i(\epsilon,x_i,\Delta\omega_0) = f_i(x_i,\Delta\omega_0) + g_i(\epsilon,x_i),\ \epsilon\in\mathbb R_{>0}$, where $g_i(\epsilon,x_i):\mathbb R_{>0}\times\mathbb R\to\mathbb R$ is twice continuously differentiable convex in $x_i$ and satisfies $\lim_{\epsilon\to+\infty} g_i(\epsilon,x_i) = \mathbb I(x_i)$. This produces a regularized solution:
    \begin{align}
        x^r &= \text{argmin}_{x\in\mathbb R^n}\sum_{i=1}^n\tilde f_i(x_i,\Delta\omega_0) \text{ s.t. (8)}.
    \end{align}
    If follows that
    \begin{align}
        \tilde f_i(\epsilon,x_i,\Delta\omega_0) \to f_i(x_i,\Delta\omega_0) + \mathbb I(x_i) \text{ as } \epsilon\to+\infty,
    \end{align}
    and consequently,
    \begin{align}\label{equivalent-1}
        \ x^r \to x^\star,\ \lambda^r \to \lambda^\star \text{ as } \epsilon\to+\infty,
    \end{align}
    where $\lambda^r$ and $\lambda^\star$ denote the optimal Lagrange multipliers associated with $x^r$ and $x^\star$, respectively. This equivalence allows us to express the dynamics of the optimal trajectory.
    
    For any $\epsilon\in(0,+\infty)$, the optimality conditions of (50) yield $0 = \tilde f_{ix}(x_i^r,\Delta\omega_0) + \lambda^r$ and $0 = \sum_{i=1}^n x_i^r - h(\Delta\omega_0)c_{agg}$. Differentiating both equations with respect to time and applying the chain rule gives $0 = \tilde f_{ixx}(x_i^r,\Delta\omega_0)\dot x_i^r + f_{ix\omega}(x_i^r,\Delta\omega_0)\Delta\dot\omega_0 + \dot\lambda^r$ and $0 = \sum_{i=1}^n \dot x_i^r - h_\omega(\Delta\omega_0)\Delta\dot\omega_0c_{agg}$. From these, we obtain
    \begin{align}
        \dot x_i^r = &-\tilde f_{ixx}(\epsilon,x_i^r,\Delta\omega_0)^{-1}\left[f_{ix\omega}(x_i^r,\Delta\omega_0)\Delta\dot\omega_0+\dot\lambda^r\right],\label{approximate optimal trajectory x}\\
        \dot\lambda^r = &-\left(\sum_{i=1}^n\tilde f_{ixx}(\epsilon,x_i^r,\Delta\omega_0)^{-1}\right)^{-1}\notag\\
        &\times\left(\sum_{i=1}^n\tilde f_{ixx}(\epsilon,x_i^r,\Delta\omega_0)^{-1}f_{ix\omega}(x_i^r,\Delta\omega_0)\Delta\dot\omega_0\right)\label{approximate optimal trajectory la}\\
        & -\left(\sum_{i=1}^n\tilde f_{ixx}(\epsilon,x_i^r,\Delta\omega_0)^{-1}\right)^{-1}h_\omega(\Delta\omega_0)\Delta\dot\omega_0c_{agg}.\notag
    \end{align}
    
    According to Lemma 1, we have $\lim_{\epsilon\to+\infty}g_{ixx}(\epsilon,x_i)=0$ if $x_i\in\mathrm{Int}(\Omega_i)$ and $\lim_{\epsilon\to+\infty}g_{ixx}(\epsilon,x_i)=+\infty$ if $x_i\in\partial\Omega_i$. Therefore, it follows from (\ref{rho}) that
    \begin{align}\label{equivalent-2}
    \lim_{\epsilon\to+\infty}\tilde f_{ixx}(\epsilon,x_i,\Delta\omega_0)^{-1}  = \rho_i.
    \end{align}
    Taking $\epsilon\to+\infty$ in (\ref{approximate optimal trajectory x})--(\ref{approximate optimal trajectory la}), substituting (\ref{equivalent-2}), and invoking (\ref{equivalent-1}), we establish the equivalence between (\ref{approximate optimal trajectory x})--(\ref{approximate optimal trajectory la}) and (\ref{optimal trajectory x})--(\ref{optimal trajectory la}). This indicates that $x_i^\star$ and $\lambda^\star$ are piecewise continuously differentiable with respect to $t$ and may fail to be differentiable only at the switching instants of $\sigma_i$. The proof is complete.

\subsection{Proof of Lemma 4}
\textbf{Part I:} We prove the feedforward invariance of $\Omega_i$ by contradiction. Suppose that the initial value satisfies $x_i(0)\in\Omega_i$, and there exists a time $t_i\geq 0$ such that $x_i(t)$ exits $\Omega_i$ at $t=t_i$. This implies that $x_i(t_i)\in\partial\Omega_i$ and $x_i(t_i+\epsilon)\notin\Omega_i$, where $\delta\to 0^+$. 1) If $e_i(t_i)=0$, then $x_i(t_i)-e_i(t_i)=x_i(t_i)\in\partial\Omega_i$, so that $\sigma_i(t_i)=0$ and $\alpha_i(t_i) = 0$. 2) If  $e_i(t_i)\neq 0$, then $\alpha_i(t_i)\text{sig}(e_i(t_i))=0$, which means $x_i$ is solely feedback-driven as it exits $\Omega_i$.

Therefore, for both cases, we have
\begin{align}
x_i(t_i+\delta) = x_i(t_i)+\delta u_i(t_i).\notag
\end{align}
Because $u_i = -f_{ixx}(x_i,\Delta\omega_0)^{-1}(\gamma_{1}\vert e_i\vert^{1-\frac{p}{q}}+\gamma_{2}\vert e_i\vert^{1 + \frac{p}{q}}+\gamma_{3})\text{sign}(e_i)$, we can rewrite $u_i(t)$ as $u_i(t)=-\eta_i(t)e_i(t)$ with a non-negative and finite $\eta_i(t)$. As a result,
\begin{align}
x_i(t_i+\delta) = \left[1-\delta\eta_i(t_i)\right]x_i(t_i) +  \delta\eta_i(t_i)\left[x_i(t_i)-e_i(t_i)\right].\notag
\end{align}
Note that $x_i(t_i)\in\partial\Omega_i$, $x_i(t_i)-e_i(t_i)=\mathcal P_{\Omega_i}(x_i(t_i)-F_i(t_i))\in\Omega_i$, and $\delta\eta_i(t_i)\in[0,1)$, it follows from convexity of $\Omega_i$ that $x_i(t_i+\delta)\in\Omega_i$. As both cases contradict the earlier assumption that $x_i(t_i+\delta)\notin\Omega_i$, it follows that $x_i(t)\in\Omega_i$ for all $t\geq 0$.

\textbf{Part II:} We show $(F_i-e_i)e_i \geq 0$ for all $t\geq 0$. According to Lemma 2, for any $y_i\in\Omega_i$, we have
\begin{align}
    &\langle (x_i-F_i)-\mathcal P_{\Omega_i}(x_i-F_i),y_i-\mathcal P_{\Omega_i}(x_i-F_i)\rangle \notag\\
    &=(e_i-F_i)(y_i-x_i+e_i)\notag\\
    &= (e_i-F_i)e_i + (e_i-F_i)(y_i-x_i) \leq 0.
\end{align}
Since $x_i\in\Omega_i$ all $t\geq 0$, as previously proven, we can simply put $y_i=x_i\in\Omega_i$, which gives $\langle e_i-F_i,y_i-x_i\rangle  = 0$ and
\begin{align}
    (e_i-F_i)e_i\leq 0.
\end{align}
Combining Parts I and II completes the proof.

\bibliographystyle{ieeetr}
\bibliography{ref}
\IEEEpubidadjcol
\newpage

\begin{IEEEbiography}
[{\includegraphics[width=1in,height=1.25in,clip,keepaspectratio]{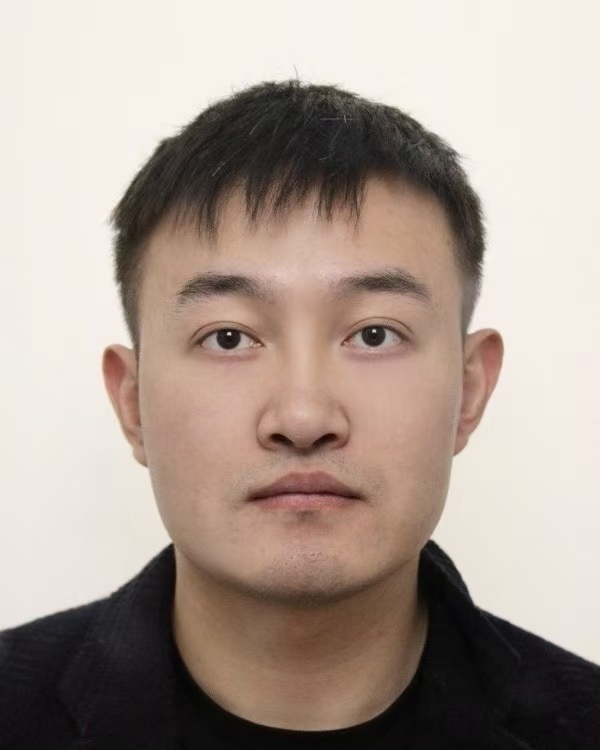}}]%
{Yiqiao Xu} received the M.Sc. degree in Advanced Control and Systems Engineering in 2018 and the Ph.D. degree in Electrical and Electronic Engineering in 2023, both from the University of Manchester, U.K. Since 2023, he has been a Postdoctoral Research Associate with the Power and Energy Division at the same institution. His research interests include distributed optimization and learning-based control of power networks and multi-energy systems. He received the Best Theory Paper Award at the IEEE 30th International Conference on Automation and Computing (ICAC) in 2025.
\end{IEEEbiography}

\begin{IEEEbiography}[{\includegraphics[width=1in,height=1.25in,clip,keepaspectratio]{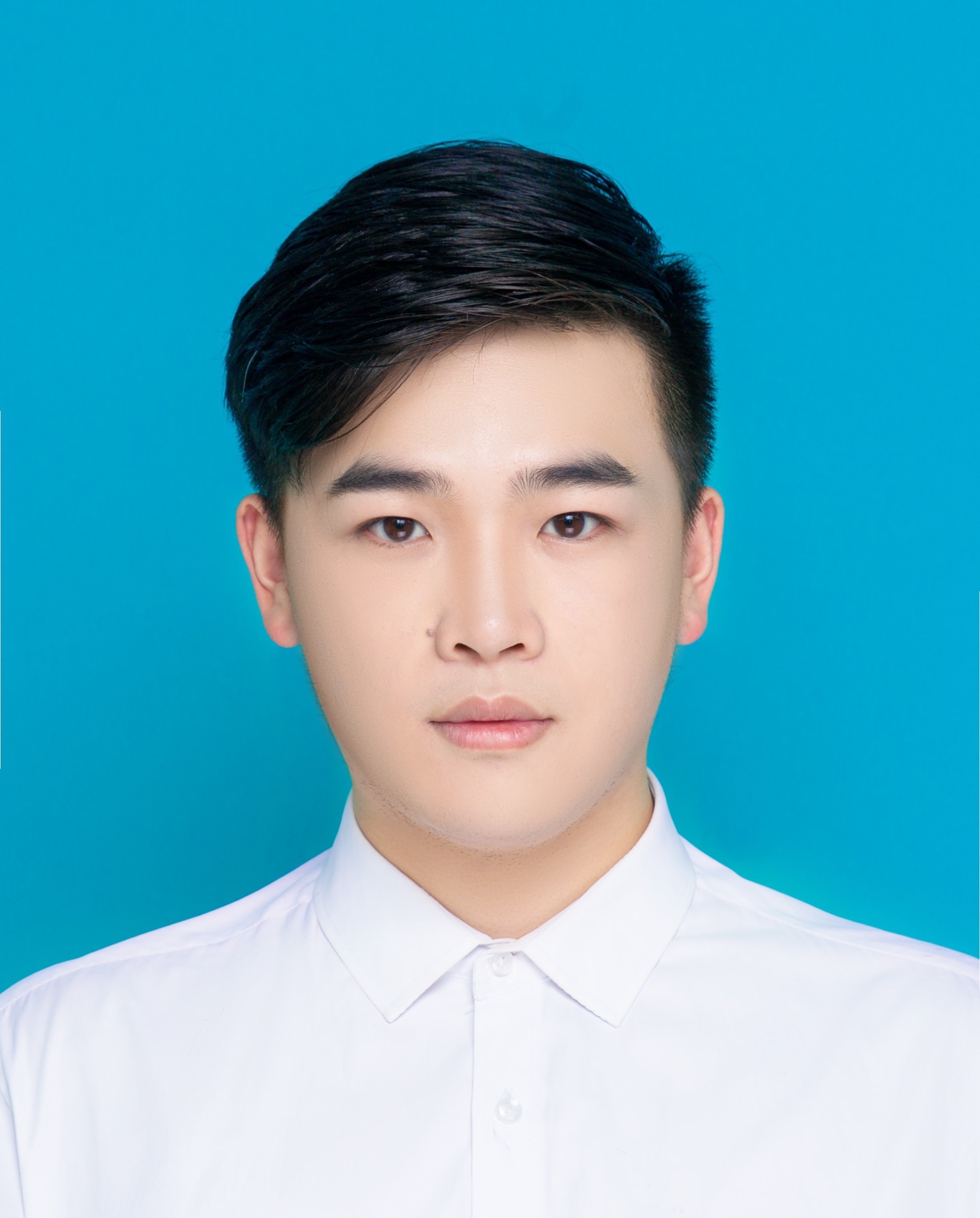}}]%
{Quan Wan} received his B.Eng. degree in vehicle engineering from Wuhan University of Technology, China, in 2016 and his M.Sc. degree in mechanical engineering from the University of Florida, USA, in 2019. He is currently pursuing a Ph.D. degree in electrical and electronic engineering at the University of Manchester, U.K. His research interests include optimal control, adaptive control, and reinforcement learning, with applications in power and energy systems. 
\end{IEEEbiography}

\begin{IEEEbiography}[{\includegraphics[width=1in,height=1.25in,clip,keepaspectratio]{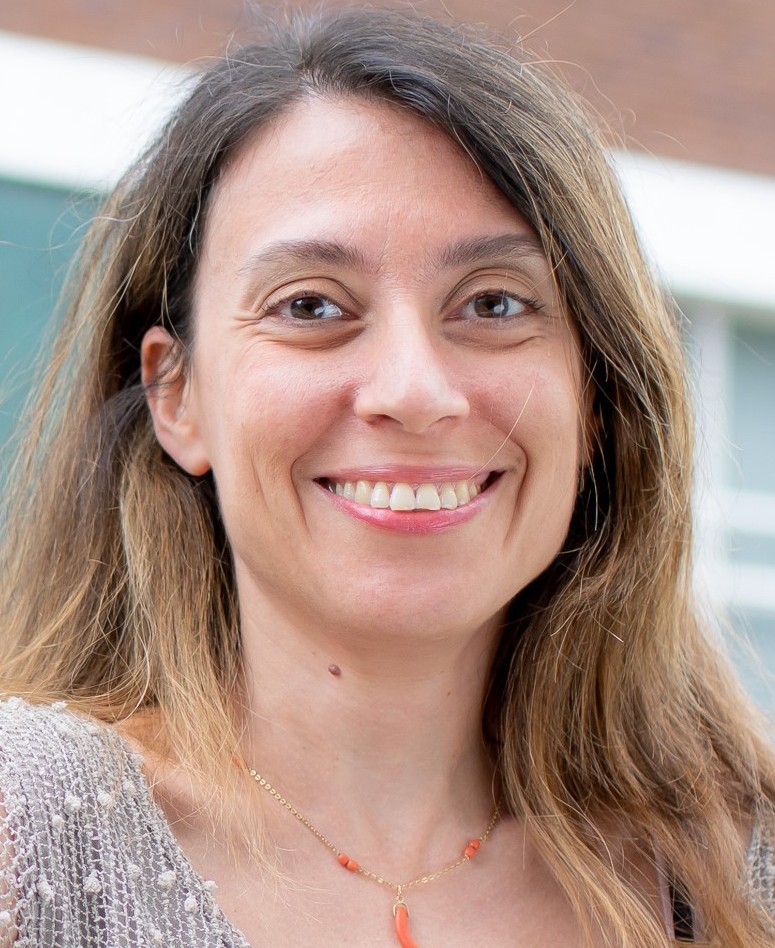}}]%
{Alessandra Parisio} is a Professor of Control of Sustainable Energy Networks, in the Department of Electrical and Electronic Engineering at The University of Manchester, U.K. She is IEEE senior member, co-chair of the IEEE RAS Technical Committee on Smart Buildings and vice-Chair for Education of the IFAC Technical Committee 9.3. Control for Smart Cities. She is currently an Associate Editor of the IEEE Transactions on Control of Network Systems, IEEE Transactions on Automation Science and Engineering, European Journal of Control and Applied Energy. Her main research interests span the areas of control engineering, in particular Model Predictive Control, distributed optimization and control, stochastic constrained control, and power systems, with energy management systems under uncertainty, optimization and control of multi-energy networks and distributed flexibility.
\end{IEEEbiography}

\end{document}